\documentclass{LMCS}

\def\doi{9(1:13)2013}
\lmcsheading%
{\doi}
{1--26}
{}
{}
{Feb.~29, 2012}
{Mar.~26, 2013}
{}
\usepackage{vaucanson-g}
\usepackage{fancyvrb}
\usepackage{todonotes}
\usepackage{url}
\usepackage{float}
\usepackage{eucal}
\usepackage{subfig}
\usepackage{changebar}
\usepackage{paralist}
\usepackage{amssymb}
\usepackage{amsmath}
\usepackage{latexsym}
\usepackage{xspace}
\usepackage[norelsize]{algorithm2e}
\usepackage[on]{auto-pst-pdf}
\usepackage{hyperref,enumerate}
\newcommand{\forappendix}[2]{#2}
\renewcommand{\forappendix}[2]{#1}

\renewcommand{\phi}{\varphi}

\newcommand\hide[1]{}
\newcommand{\zug}[1]{{\langle #1  \rangle}}
\newcommand{\set}[1]{\text{$\{#1\}$}} 
\newcommand{\rzug}[1]{{\scriptstyle\langle} #1  {\scriptstyle\rangle}}

\newcommand\kahler{K\"ahler\xspace}
\newcommand\konig{Ko\"nig\xspace}
\newcommand\buchi{B\"uchi\xspace}
\newcommand{\N}{\mbox{I$\!$N}}

\newcommand\A{{\mathcal A}}
\newcommand{\B}{{\mathcal B}}

\newcommand\splus{\ensuremath{+\!}}

\newcommand\mcS{\mathcal{S}}

\newcommand{\DAG}{\textsc{dag}\xspace}
\newcommand{\DAGs}{\textsc{dag}s\xspace}

\def\squarebox#1{\hbox to #1{\hfill\vbox to #1{\vfill}}}
\renewcommand{\qed}{\hspace*{\fill}
          \vbox{\hrule\hbox{\vrule\squarebox{.667em}\vrule}\hrule}\smallskip}
\newcommand{\abs}[1]{\lvert#1\rvert}

\nochangebars

\newcommand{\standout}[1]{\medskip \noindent {\bf #1:}}

\newcounter{examplectr}
\newcommand{\G}{\ensuremath{G}\xspace}
\newcommand{\Gprime}{\ensuremath{G'}\xspace}
\newcommand{\Gdubprime}{\ensuremath{G''}\xspace}

\newcommand{\LR}{\ensuremath{\mathcal{R}}}

\newcommand{\TLRM}{\ensuremath{\mathcal{R}^m_T}}
\newcommand{\PSQ}{\ensuremath{{\bf Q}}}

\SetKwFunction{torank}{torank}
\SetKwFunction{pred}{pred}
\SetKwFunction{maxrank}{mr}
\SetKwFunction{tighten}{tighten}
\SetKwFunction{min}{min}
\SetKwFunction{prerank}{\ensuremath{\alpha}}
\SetKwFunction{prerankstate}{\ensuremath{\beta}}

\newtheorem{theorem}{Theorem}[section]

\newtheorem{corollary}[theorem]{Corollary}

\begin{document}
\title[Unifying B\"uchi Complementation Constructions]{Unifying B\"uchi Complementation Constructions}

\author[S.~Fogarty]{Seth Fogarty\rsuper a}
\address{{\lsuper a}Computer Science Department, Trinity University, San Antonio, TX}
\email{sfogarty@trinity.com}
\thanks{{\lsuper a}The authors are grateful to Yoad Lustig for his extensive help in analyzing the original
slice-based construction.  Work supported in part by NSF grants CNS-1049862 and CCF-1139011, by NSF
Expeditions in Computing project "ExCAPE: Expeditions in Computer Augmented Program Engineering," by
BSF grant 9800096, and by gift from Intel. Work by Seth Fogarty done while at Rice University.}
\author[O.~Kupferman]{Orna Kupferman\rsuper b}
\address{{\lsuper b}School of Computer Science and Engineering, Hebrew University of Jerusalem, Israel}
\email{orna@cs.huji.ac.il}
\author[T.~Wilke]{Thomas Wilke\rsuper c}
\address{{\lsuper c}Institut f\"ur Informatik, Christian-Albrechts-Universit\"at zu Kiel, Kiel, Germany}
\email{wilke@ti.informatik.uni-kiel.de}
\author[M.~Y.~Vardi]{Moshe Y.~Vardi\rsuper d}
\address{{\lsuper d}Department of Computer Science, Rice University, Houston, TX}
\email{vardi@cs.rice.edu}

\ACMCCS{[{\bf Theory of computation}]:  Formal languages and automata
  theory---Automata over infinite objects; [{\bf Software and its
      engineering}]: Software organization and properties---Software functional properties---Formal methods---Model checking} 

\begin{abstract}
Complementation of \buchi automata, required for checking automata containment, is of major
theoretical and practical interest in formal verification. We consider two recent approaches to
complementation. The first is the {\em rank-based approach} of Kupferman and Vardi, which operates
over a \DAG that embodies all runs of the automaton. This approach is based on the observation that
the vertices of this \DAG can be ranked in a certain way, termed an {\em odd ranking}, iff all 
runs are rejecting. The second is the {\em slice-based approach} of \kahler and Wilke. This approach
tracks levels of ``split trees'' -- run trees in which only essential information
about the history of each run is maintained. While the slice-based construction is conceptually
simple, the complementing automata it generates are exponentially larger than those of the recent
rank-based construction of Schewe, and it suffers from the difficulty of symbolically encoding
levels of split trees.

In this work we reformulate the slice-based approach in terms of run \DAGs and preorders over
states. In doing so, we begin to draw parallels between the rank-based and slice-based approaches.
Through deeper analysis of the slice-based approach, we strongly restrict the nondeterminism it
generates. We are then able to employ the slice-based approach to provide a new odd ranking, called
a {\em retrospective ranking}, that is different from the one provided by Kupferman and Vardi.  This
new ranking allows us to construct a deterministic-in-the-limit rank-based automaton with a highly
restricted transition function.  Further, by phrasing the slice-based approach in terms of ranks,
our approach affords a simple symbolic encoding and achieves the tight bound of Schewe's
construction.
\end{abstract}

\maketitle

\section{Introduction}
The complementation problem for nondeterministic automata is central to the automata-theoretic
approach to formal verification \cite{Var07a}.  To test that the language of an automaton $\A$ is
contained in the language of a second automaton $\B$, check that the intersection of $\A$ with an
automaton that complements ${\B}$ is empty.  In model checking, the automaton ${\A}$ corresponds to
the system, and the automaton ${\B}$ corresponds to a property \cite{VW86b}.  While it is easy to
complement properties given as temporal logic formulas, complementation of properties given as
automata is not simple. Indeed, a word $w$ is rejected by a nondeterministic automaton $\A$ if
\emph{all} runs of $\A$ on $w$ reject the word. Thus, the complementary automaton has to consider
all possible runs, and complementation has the flavor of determinization.  Representing liveness,
fairness, or termination properties requires automata that recognize languages of infinite words.
Most commonly considered are nondeterministic \buchi automata, in which some of the states are
designated as accepting, and a run is accepting if it visits accepting states infinitely often
\cite{Buc62}.  For automata on finite words, determinization, and hence also complementation, is
done via the subset construction \cite{RS59}.  For B\"uchi automata the subset construction is not
sufficient, and optimal complementation constructions are more complicated \cite{Var07b}.

Efforts to develop simple complementation constructions for \buchi automata started early
in the 60s, motivated by decision problems of second-order logics.  \buchi suggested a
complementation construction for nondeterministic \buchi automata that involved a 
Ramsey-based combinatorial argument and a doubly-exponential blow-up in the state space
\cite{Buc62}.  Thus, complementing an automaton with $n$ states resulted in an automaton with
$2^{2^{O(n)}}$ states.  In \cite{SVW87}, Sistla et al. suggested an improved implementation of
B\"uchi's construction, with only $2^{O(n^2)}$ states, which is still not optimal.  Only
in \cite{Saf88} Safra introduced a determinization construction, based on {\em Safra trees}, which also
enabled a $2^{O(n \log n)}$ complementation construction, matching a lower bound described by Michel
\cite{Mic88}.  A careful analysis of the exact blow-up in Safra's and Michel's bounds, however,
reveals an exponential gap in the constants hiding in the $O()$ notations: while the upper bound on
the number of states in the complementary automaton constructed by Safra is $n^{2n}$, Michel's lower
bound involves only an $n!$ blow up, which is roughly $(n/e)^n$. In addition, Safra's construction
has been resistant to optimal implementations \cite{ATW06,THB95}, which has to do with the complicated
combinatorial structure of its states and transitions, which can not be encoded symbolically.

The use of complementation in practice has led to a resurgent interest in the exact blow-up that
complementation involves and the feasibility of a symbolic complementation construction.  In 2001,
Kupferman and Vardi suggested a new analysis of runs of \buchi automata that led to a simpler
complementation construction \cite{KV01c}. In this analysis, one considers a \DAG that
embodies all the runs of an automaton $\A$ on a given word $w$. It is shown in \cite{KV01c} that the
nodes of this \DAG can be mapped to ranks, where the rank of a node essentially indicates the
progress made towards a suffix of the run with no accepting states.
Further, all the runs of $\A$ on $w$ are rejecting iff there is a {\em bounded odd ranking\/} of the
\DAG: one in which the maximal rank is bounded, ranks along paths do not increase, paths become
trapped in odd ranks, and nodes associated with accepting states are not assigned an odd rank.
Consequently, complementation can circumvent Safra's determinization construction along with the
complicated data structure of Safra trees, and can instead be based on an automaton that guesses an
odd ranking. The state space of such an automaton is based on annotating states in subsets with 
the guessed ranks. Beyond the fact that the {\em rank-based construction\/} can be implemented
symbolically \cite{TV07}, it gave rise to a sequence of works improving both the blow-up it involves
and its implementation in practice.  The most notable improvements are the introduction of tight
rankings \cite{FKV06} and Schewe's improved cut-point construction \cite{Sch09}. These improvements tightened
the $(6n)^n$ upper bound of \cite{KV01c} to $(0.76n)^n$.
Together with recent work on a tighter lower bound \cite{Yan06}, the gap between the upper and lower
bound is now a quadratic term.  Addressing practical concerns, Doyen and Raskin have introduced a
useful subsumption technique for the rank-based approach \cite{DR09}.

In an effort to unify \buchi complementation with other operations on automata,  \kahler and Wilke
introduced yet another analysis of runs of nondeterministic B\"uchi automata \cite{KW08}. The
analysis is based on {\em reduced split trees}, which are related to the M\"uller-Schupp trees used
for determinization \cite{MS95}.  A reduced split tree is a binary tree whose nodes are sets of
states as follows: the root is the set of initial states; and given a node associated with a set of
states, its left child is the set of successors that are accepting, while the right child is the set
of successors that are not accepting. In addition, each state of the automaton appears at most once
in each level of the binary tree: if it would appear in more than one set, it occurs only in the
leftmost one.  The construction that follows from the analysis, termed the {\em slice-based
construction\/},  is simpler than Safra's determinization, but its implementation suffers from
similar  difficulties: the need to refer to leftmost children requires encoding of a preorder, and
working with reduced split trees makes the transition relation between states awkward.  Thus,
as has been the case with Safra's construction, it is not clear how the slice-based approach can be
implemented symbolically.  This is unfortunate, as the slice-based approach does offer a very clean
and intuitive analysis, suggesting that a better construction is hidden in it.

In this paper we reveal such a hidden, elegant, construction, and we do so by unifying the
rank-based and the slice-based approaches. Before we turn to describe our construction, let us point
to a key conceptual difference between the two approaches. This difference has made their relation
of special interest and challenge.  In the rank-based approach, the ranks assigned to a node bound
the visits to accepting states yet to come. Thus, the ranks refer to the {\em future\/} of the run,
making the rank-based approach inherently nondeterministic. In contrast, in the slice-based
approach, the partition of the states of the automaton to the different sets in the tree is based on
previous visits to accepting states. Thus, the partition refers to the {\em past\/} of the run, and
does not depend on its future.

In order to draw parallels between the two approaches, we present a formulation of the slice-based
approach \hide{of \kahler and Wilke} in terms of run \DAGs. A careful analysis of the slice-based
approach then enables us to reduce the nondeterminism in the construction.  We can then employ this
improved slice-based approach in order to define a particular odd ranking of rejecting run \DAGs,
called a {\em retrospective ranking}.  In addition to revealing the theoretical connections between
the two seemingly different approaches, the new ranks lead to a complementation construction with a
transition function that is smaller and deterministic in the limit: every accepting run of the
automaton is eventually deterministic. This presents the first 
deterministic-in-the-limit complementation construction that does not use determinization.  Determinism in the limit
is central to verification in probabilistic settings \cite{CY95} and has proven useful in
experimental results \cite{ST03}.  Phrasing slice-based complementation as an odd ranking also
immediately affords us the improved cut-point of Schewe, the subsumption operation of Doyen and
Raskin, and provides an easy symbolic encoding.

\section{Preliminaries}\label{Sect:Ranks}\label{Sect:Tight}


A \emph{nondeterministic \buchi automaton on infinite words} (NBW for short) is a tuple $\A=\zug{\Sigma, Q,
Q^{in}, \rho, F}$, where $\Sigma$ is a finite alphabet, $Q$ a finite set of states, $Q^{in}
\subseteq Q$ a set of initial states, $F \subseteq Q$ a set of accepting states, and $\rho \colon Q
\times \Sigma \to 2^Q$ a nondeterministic transition relation. A state $q \in Q$ is {\em
deterministic} if for every $\sigma \in \Sigma$ it holds that $\abs{\rho(q,\sigma)} \leq 1$. We lift
the function $\rho$ to sets $R$ of states in the usual fashion: $\rho(R,\sigma) = \bigcup_{q \in R}
\rho(q,\sigma)$. Further, we lift $\rho$ to words word $\sigma_0\cdots \sigma_i$ by defining
$\rho(R,\sigma_0\cdots\sigma_i)$ = $\rho(\rho(R,\sigma_0),\sigma_1\cdots\sigma_i)$. For
completeness, let $\rho(R,\epsilon)=R$.

\cbstart An infinite {\em run} of an NBW $\A$ on a word $w=\sigma_0\sigma_1\cdots \in \Sigma^\omega$
is an infinite sequence of states $p_0,p_1,\ldots\in Q^\omega$ such that $p_0 \in Q^{in}$ and, for
every $i \geq 0$, we have $p_{i+1} \in \rho(p_i, \sigma_i)$. Correspondingly, a \emph{finite run} is a
finite sequence of states $p_0,\ldots,p_n$ such that $p_0 \in Q^{in}$ and, for every $0 \leq i <
n$, we have $p_{i+1} \in \rho(p_i, \sigma_i)$. When unspecified, a run refers to an infinite run. 
\cbend  A run is \emph{accepting} iff $p_i \in F$ for infinitely many $i \in \N$.  A word $w \in
\Sigma^\omega$ is accepted by $\A$ if there is an accepting run of $\A$ on $w$.  The words accepted
by $\A$ form the {\em language} of $\A$, denoted by $L(\A)$. The complement of $L(\A)$, denoted
$\overline{L(\A)}$, is $\Sigma^\omega \setminus L(\A)$.  We say an automaton is \emph{deterministic
in the limit} if every state reachable from an accepting state is deterministic. Converting $\A$ to
an equivalent deterministic in the limit automaton involves an exponential blowup \cite{CY95,Saf88}.
One can simultaneously complement and determinize in the limit, via co-determinization into a parity
automaton \cite{Pit06}, and then converting that parity automaton to a deterministic-in-the-limit
\buchi automaton, with a cost of $(n^2/e)^n$. 

\standout{Run \DAGs} 
Consider an NBW $\A$ and an infinite word $w = \sigma_0\sigma_1\cdots$. The runs of $\A$ on $w$ can
be arranged in an infinite \DAG (directed acyclic graph) $\G=\zug{V,E}$, where 
\begin{iteMize}{$\bullet$}
\cbstart
\item $V \subseteq Q \times \N$ is such that $\rzug{q,i} \in V$ iff some finite or infinite run $p$ of $\A$ on $w$ has
$p_i=q$.
\cbend
\item $E \subseteq {\displaystyle{\bigcup}_{i \geq 0}}~(Q \!\times\! \set{i}) \times (Q \!\times\!\set{i\splus 1})$ is
s.t.  $E(\rzug{q,i},\rzug{q',i \splus 1})$ iff $\rzug{q,i} \in V$ and $q' \in \rho(q, \sigma_i)$.
\end{iteMize}
The \DAG $\G$, called the \emph{run \textsc{dag} of $\A$ on $w$}, embodies all possible runs
of $\A$ on $w$. We are primarily concerned with \emph{initial paths} in $\G$: paths that start in
$Q^{in} \times \set{0}$.  Define a node $\rzug{q,i}$ to be an $F$-node when $q \in F$, and a path in
$\G$ to be \emph{accepting} when it is both initial and contains infinitely many $F$-nodes.  An
accepting path in $\G$ corresponds to an accepting run of $\A$ on $w$. When $\G$ contains an accepting path, 
call $\G$ an accepting run \DAG, otherwise call it a rejecting run \DAG.  We often consider
\DAGs $H$ that are subgraphs of $\G$. A node $u$ is a \emph{descendant} of $v$ in $H$ when $u$ is
reachable from $v$ in $H$. A node $v$ is \emph{finite} in $H$ if it has only finitely many
descendants in $H$.  A node $v$ is \emph{$F$-free} in $H$ if it is not an $F$-node, and has no
descendants in $H$ that are $F$-nodes. We say a node \emph{splits} when it has at least two
children, and conversely that two nodes \emph{join} when they share a common child.

\begin{exa}\label{ExampleOne}
In Figure~\ref{Fig:Automaton} we describe an NBW $\A$ that accepts
words with finitely many letters
$b$. On the right is a prefix of the rejecting run \DAG of $\A$ on $w=babaabaaabaaaa\cdots$\hide{,
the word with an increasing number of $a$'s between successive $b$'s}. 
\end{exa}

{
\begin{figure}
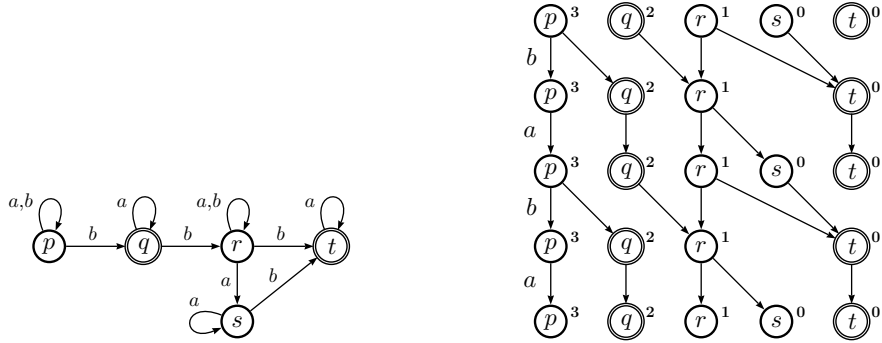

\begin{centering}
\subfloat{
\ChgEdgeLabelScale{0.8}
{
\begin{postscript}
\SmallPicture\VCDraw{\begin{VCPicture}{(0,-3)(10,5)}
\State[p]{(0,0)}{p} \FinalState[q]{(2.5,0)}{q} \State[r]{(5,0)}{r}
\State[s]{(5,-2)}{s} \FinalState[t]{(7.5,0)}{t}
\EdgeL{p}{q}{b}
\EdgeL{q}{r}{b}
\EdgeL{r}{t}{b}
\EdgeR{r}{s}{a}
\EdgeL{s}{t}{b}

\CLoopN{p}{\tiny{a,\!b}}
\CLoopN{q}{a}
\CLoopN{r}{a,\!b}
\CLoopW{s}{a}
\CLoopN{t}{a}

\end{VCPicture}}
\end{postscript}
}

\quad}\quad\quad
\subfloat{
{
\begin{postscript}
\SmallPicture\VCDraw{\begin{VCPicture}{(0,-12)(8,2)}
\def\level{0}\def\offset{0}
\State[p]{(0,\offset)}{p\level} \FinalState[q]{(2,\offset)}{q\level} \State[r]{(4,\offset)}{r\level}
\State[s]{(6,\offset)}{s\level} \FinalState[t]{(8,\offset)}{t\level}
\def\level{1} \def\prevlevel{0} \def\offset{-2}
\State[p]{(0,\offset)}{p\level} \FinalState[q]{(2,\offset)}{q\level} \State[r]{(4,\offset)}{r\level}
\FinalState[t]{(8,\offset)}{t\level}
\EdgeR{p\prevlevel}{p\level}{b~~~~~~~} \EdgeL{r\prevlevel}{r\level}{} \EdgeL{p\prevlevel}{q\level}{}
\EdgeL{q\prevlevel}{r\level}{} \EdgeL{r\prevlevel}{t\level}{} \EdgeL{s\prevlevel}{t\level}{}
\def\prevlevel{1} \def\level{2}\def\offset{-4}
\State[p]{(0,\offset)}{p\level} \FinalState[q]{(2,\offset)}{q\level} \State[r]{(4,\offset)}{r\level}
\State[s]{(6,\offset)}{s\level} \FinalState[t]{(8,\offset)}{t\level}
\EdgeR{p\prevlevel}{p\level}{a~~~~~~~} \EdgeL{r\prevlevel}{r\level}{} \EdgeL{q\prevlevel}{q\level}{}
\EdgeL{t\prevlevel}{t\level}{} \EdgeL{r\prevlevel}{s\level}{} 
\def\prevlevel{2} \def\level{3}\def\offset{-6}
\State[p]{(0,\offset)}{p\level} \FinalState[q]{(2,\offset)}{q\level} \State[r]{(4,\offset)}{r\level}
\FinalState[t]{(8,\offset)}{t\level}
\EdgeR{p\prevlevel}{p\level}{b~~~~~~~} \EdgeL{r\prevlevel}{r\level}{} \EdgeL{p\prevlevel}{q\level}{}
\EdgeL{q\prevlevel}{r\level}{} \EdgeL{r\prevlevel}{t\level}{} \EdgeL{s\prevlevel}{t\level}{}
\def\prevlevel{3} \def\level{4}\def\offset{-8}
\State[p]{(0,\offset)}{p\level} \FinalState[q]{(2,\offset)}{q\level} \State[r]{(4,\offset)}{r\level}
\State[s]{(6,\offset)}{s\level} \FinalState[t]{(8,\offset)}{t\level}
\EdgeR{p\prevlevel}{p\level}{a~~~~~~~} \EdgeL{r\prevlevel}{r\level}{} \EdgeL{q\prevlevel}{q\level}{}
\EdgeL{t\prevlevel}{t\level}{} \EdgeL{r\prevlevel}{s\level}{} 
\hide{
\def\prevlevel{4} \def\level{5}\def\offset{-10}
\State[p]{(0,\offset)}{p\level} \FinalState[q]{(2,\offset)}{q\level} \State[r]{(4,\offset)}{r\level}
\State[s]{(6,\offset)}{s\level} \FinalState[t]{(8,\offset)}{t\level}
\EdgeR{p\prevlevel}{p\level}{a~~~~~~~} \EdgeL{r\prevlevel}{r\level}{} \EdgeL{q\prevlevel}{q\level}{}
\EdgeL{s\prevlevel}{s\level}{} \EdgeL{t\prevlevel}{t\level}{} \EdgeL{r\prevlevel}{s\level}{} 
\def\prevlevel{5} \def\level{6}\def\offset{-12}
\State[p]{(0,\offset)}{p\level} \FinalState[q]{(2,\offset)}{q\level} \State[r]{(4,\offset)}{r\level}
\FinalState[t]{(8,\offset)}{t\level}
\EdgeR{p\prevlevel}{p\level}{b~~~~~~~} \EdgeL{r\prevlevel}{r\level}{} \EdgeL{p\prevlevel}{q\level}{}
\EdgeL{q\prevlevel}{r\level}{} \EdgeL{r\prevlevel}{t\level}{} \EdgeL{s\prevlevel}{t\level}{}
}
\begin{boldmath}
\multirput(0.65,0.25)(0.0,-2.0){5}{\large$3$}
\multirput(2.65,0.25)(0.0,-2.0){5}{\large$2$}
\multirput(4.65,0.25)(0.0,-2.0){5}{\large$1$}
\multirput(6.65,0.25)(0.0,-4.0){2}{\large$0$}
\multirput(6.65,-7.75)(0.0,-2.0){1}{\large$0$}
\multirput(8.65,0.25)(0.0,-2.0){5}{\large$0$}

\end{boldmath}

\end{VCPicture}}
\end{postscript}
}
}\quad
\vspace{-0.5in}
\caption{Left, the NBW $\A$, in which all states are initial.
Right, the rejecting run \DAG $\G$ of $\A$ on $w=babaabaaabaaaa\cdots$.
Nodes are superscripted with the prospective ranks of Section~\ref{Sect:Ranks}.
}\label{Fig:Automaton}
\end{centering}
\end{figure}
}


\standout{Rank-Based Complementation} 
If an NBW $\A$ does not accept a word $w$, then every run of $\A$ on $w$ must eventually cease
visiting accepting states. The notion of \emph{ranking}s, foreshadowed in \cite{Kla90} and
introduced in \cite{KV01c}, uses natural numbers to track the progress of each run in the \DAG
towards this point. A ranking for a \DAG $\G=\zug{V,E}$ is a mapping from $V$ to $\N$,
in which no $F$-node is given an odd rank, and in which the ranks along all paths do not increase.
Formally, a ranking is a function ${\bf r} \colon V \to \N$ such that if $u \in V$ is an $F$-node
then ${\bf r}(u)$ is even; and for every $u,v \in V$, if $(u,v) \in E$ then ${\bf r}(u) \geq {\bf
r}(v)$. Since each path starts at a finite rank and ranks cannot increase, every path eventually
becomes trapped in a rank.  A ranking is called an \emph{odd ranking} if every path becomes trapped
in an odd rank.  Since $F$-nodes cannot have odd ranks, if there exists an odd ranking ${\bf r}$, then
every path in $\G$ must stop visiting accepting nodes when it becomes trapped in its final, odd,
rank, and $\G$ must be a rejecting \DAG.

\begin{lem}\label{Odd_Ranking_Rejecting}{\rm\cite{KV01c}}
If a run \DAG $\G$ has an odd ranking, then $\G$ is rejecting.
\end{lem}

A ranking is \emph{bounded by $l$} when its range is $\set{0,...,l}$, and an NBW $\A$ is of rank $l$
when for every $w \not\in L(\A)$, the rejecting \DAG $\G$ has an odd ranking bounded by $l$. If we
can prove that an NBW $\A$ is of rank $l$, we can use the notion of odd rankings to construct a
complementary automaton. This complementary NBW, denoted $\A^l_R$, tracks the levels of the run \DAG
and attempts to guess an odd ranking bounded by $l$.  An \emph{$l$-bounded level ranking} for an NBW
$\A$ is a function $f \colon Q \to \set{0,\ldots,l,\bot}$, such that if $q \in F$ then $f(q)$ is
even or $\bot$.  Let ${\LR}^l$ be the set of all $l$-bounded level rankings.  The state space of
$\A^l_R$ is based on the set of $l$-bounded level rankings for $\A$.  To define transitions of
$\A^l_R$, we need the following notion: for $\sigma \in \Sigma$ and $f, f' \in {\LR}^l$, say that
\emph{$f'$ follows $f$ under $\sigma$} when for every $q \in Q$ and $q' \in \rho(q,\sigma)$, if $f(q) \neq
\bot$ then $f'(q') \neq \bot$ and $f'(q') \leq f(q)$: i.e. no transition between $f$ and $f'$ on
$\sigma$ increases in rank.  Finally, to ensure that the guessed ranking is an odd ranking, we
employ the cut-point construction of Miyano and Hayashi, which maintains an obligation set of nodes
along paths obliged to visit an odd rank \cite{MH84}. For a level ranking $f$, let
$even(f)=\set{q\mid f(q)\text{ is even}}$ and $odd(f)=\set{q\mid f(q)\text{ is odd}}$.

\begin{defi}\label{KVDef}
For an NBW  $\A = \zug{\Sigma, Q, Q^{in}, \rho, F}$ and $l \in \N$, define
$\A^l_R$ to be the NBW $\zug{\Sigma, {\LR}^l \times 2^Q, \rzug{f^{in},\emptyset},
\rho_R, {\LR}^l \times \set{\emptyset}}$, where
\smallskip
\begin{iteMize}{$\bullet$}
\item $f^{in}(q)=l$ for each $q \in Q^{in}$,\ $\bot$ otherwise.
\smallskip
\item 
$\rho_R(\rzug{f,O},\sigma) = 
\begin{cases}
\{\rzug{f',~\rho(O,\sigma) \setminus odd(f')} \mid 
      f'\text{ follows $f$ under $\sigma$}\} & \text{if }O \neq \emptyset,\\
\{\rzug{f',~even(f')} \mid 
      f'\text{ follows $f$ under $\sigma$}\} & \text{if }O=\emptyset.\\
\end{cases}$
\end{iteMize}
\end{defi}


By \cite{KV01c}, for every $l \in \N$, the NBW $\A_R^l$ accepts only words rejected by $\A$ ---
exactly all words for which there exists an odd ranking with maximal rank $l$. In addition,
\cite{KV01c} proves that for every rejecting run \DAG there exists a bounded odd ranking.  Below we
sketch the derivation of this ranking. Given a rejecting run \DAG $\G$, we inductively define a
sequence of subgraphs by eliminating nodes that cannot be part of accepting runs.  At odd steps we
remove finite nodes, while in even steps we remove nodes that are $F$-free.  Formally, define a
sequence of subgraphs as follows:

\smallskip
\begin{iteMize}{$\bullet$}
\item $\G_0=\G$.
\smallskip
\item $\G_{2i+1}=\G_{2i} \setminus \set{v\mid v\text { is finite in }\G_{2i}}$.
\smallskip
\item $\G_{2i+2}=\G_{2i+1} \setminus \set{v\mid v\text { is $F$-free in }\G_{2i+1}}$.
\end{iteMize}
\medskip

It is shown in \cite{GKSV03,KV01c} that only $m=2\abs{Q\setminus F}$ steps are necessary to remove all
nodes from a rejecting run \DAG: $\G_m$ is empty. Nodes can be ranked by the last graph in which
they appear: for every node $u \in \G$, the \emph{prospective rank} of $u$ is the index $i$ such
that $u \in \G_i$ but $u \not\in \G_{i+1}$.  The \emph{prospective ranking} of $\G$ assigns every
node its prospective rank.  Paths through $\G$ cannot increase in prospective rank, and no $F$-node
can be given an odd rank: thus the prospective ranking abides by the requirements for rankings. We
call these rankings prospective because the rank of a node depends solely on its descendants.  By
\cite{KV01c}, if $\G$ is a rejecting run \DAG, then the prospective ranking of $\G$ is an odd
ranking bounded by $m$. By the above, we thus have the following. 

\begin{theorem}\label{KV_Complement}{\rm \cite{KV01c}}
For every NBW $\A$, it holds that $L(\A^m_R)=\overline{L(\A)}$.
\end{theorem}

\begin{exa}\label{Example2}
In Figure~\ref{Fig:Automaton}, nodes for states $s$ and $t$ are finite in $\G_0$.
With these nodes removed, $r$-nodes are $F$-free in $\G_1$.  Without $r$-nodes, $q$-nodes 
are finite in $\G_2$.  Finally, $p$-nodes are $F$-free in $\G_3$.
\end{exa}

Karmarkar and Chakraborty have derived both theoretical and practical benefits from exploiting properties
of this prospective ranking: they demonstrated an unambiguous complementary automaton that, for
certain classes of problems, is exponentially smaller than $\A^m_R$ \cite{KC09}.


\standout{Tight Rankings} 
For an odd ranking $\bf{r}$ and $l \in \N$, let ${\it max\_rank}({\bf
r},l)$ be the maximum rank that ${\bf r}$ assigns a vertex on level $l$ of the run \DAG.  We say that $\bf{r}$ is
{\em tight}\footnote{This definition of tightness for an odd ranking is weaker that of
\cite{FKV06}, but does not affect the resulting bounds.} if there exists an $i \in \N$ such that, for
every level $l \geq i$, all odd ranks below ${\it max\_rank}({\bf r},l)$ appear on level $l$.  It is
shown in \cite{FKV06} that the retrospective ranking is tight.  This observation suggests two
improvements to $\A^m_R$.  First, we can postpone, in an unbounded manner, the level in which it
starts to guess the level ranking. Until this point, $\A^m_R$ may use sets of states to
deterministically track only the levels of the run \DAG, with no attempt to guess the ranks.
Second, after this point, $\A^m_R$ can restrict attention to {\em tight level rankings} -- ones in
which all the odd ranks below the maximal rank appear.  Formally, say a level ranking $f$ with a
maximum rank $max\_rank = \text{max}\set{f(q)~|~q \in Q,~f(q) \neq \bot}$ is tight when, for every odd $i
\leq max\_rank$, there exists a $q \in Q$ such that $f(q)=i$. Let $\TLRM$ be the subset of $\LR^m$ that
contains only tight level rankings. The size of $\TLRM$ is at most $(0.76n)^n$ \cite{FKV06}.
Including the cost of the cut-point construction, this reduces the state space of $\A^m_R$ to
$(0.96n)^n$.

\section{Analyzing \DAGs With Profiles}\label{Sect:Comp_Profiles}

In this section we present an alternate formulation of the slice-based complementation construction
of \kahler and Wilke \cite{KW08}. Whereas \kahler and Wilke approached the problem through
reduced split trees, we derive the slice-based construction directly from an analysis of
the run \DAG. This analysis proceeds by pruning $\G$ in two steps: the first removes edges, and the
second removes vertices.

\subsection{Profiles}
Consider a run \DAG $\G=\zug{V,E}$. Let the labeling function $\Lambda \colon V \to \set{0,1}$ be
such that $\Lambda(\rzug{q,i}) = 1$ if $q \in F$ and $\Lambda(\rzug{q,i}) = 0$ otherwise. Thus,
$\Lambda$ labels $F$-nodes by $1$ and all other nodes by $0$. The \emph{profile} of a path in $\G$
is the sequence of labels of nodes in the path. The profile of a node is then the lexicographically
maximal profile of all initial paths to that node.  Formally, let $\leq$ be the lexicographic
ordering on $\set{0,1}^* \cup \set{0,1}^\omega$.  The profile of a finite path
$b=v_0,v_1,\ldots,v_n$ in $\G$, written $h_b$, is $\Lambda(v_0)\Lambda(v_1)\cdots \Lambda(v_n)$, and
the profile of an infinite path $b=v_0,v_1,\ldots$ is $h_b=\Lambda(v_0)\Lambda(v_1)\cdots$. Finally,
the profile of a node $v$, written $h_v$, is the lexicographically maximal element of $\set{h_b\mid
b \text{ is an initial path to }v}$. The lexicographic order of profiles induces a preorder over
nodes.

We define the sequence of preorders $\preceq_i$ over the nodes on each level of the run \DAG as
follows.  For every two nodes $u$ and $v$ on a level $i$, we have that $u \prec_i v$ if
$h_u < h_v$, and $u \approx_i v$ if $h_u = h_v$. For convenience, we conflate nodes on the $i$th level of
the run \DAG with their states when employing this preorder, and say $q \preceq_i r$ when $\zug{q,i}
\preceq_i \zug{r,i}$. Note that $\approx_i$ is an equivalence relation. Since the final element of a
node's profile is $1$ iff the node is an $F$-node, all nodes in an equivalence class must agree
on membership in $F$. We call an equivalence class an $F$-class when all its members are $F$-nodes,
and a non-$F$-class when none of its members is an $F$-node.  We now use profiles in order to remove
from $\G$ edges that are not on lexicographically maximal paths.  Let $\Gprime$ be the subgraph of
$\G$ obtained by removing all edges $\rzug{u,v}$ for which there is another edge $\rzug{u',v}$ such
that $u \prec_{\abs{u}} u'$.  Formally, $\Gprime=\zug{V,E'}$ where $E' = E \setminus
\set{\rzug{u,v}\mid \text{there exists }u' \in V \text{ such that }\rzug{u',v} \in E \text{ and } u
\prec_{\abs{u}} u'}$.

\begin{lem}\label{Gprime_Captures_Profiles}
For every two nodes $u$ and $v$, if $(u,v) \in E'$, then $h_{v} \in \set{h_u0, h_u1}$.
\end{lem}
\begin{proof}
Assume by way of contradiction that $h_{v} \not \in \set{h_u0,h_u1}$.  Recall that $h_v$ is the
lexicographically maximal element of $\set{h_b\mid b \text{ is an initial path to }v}$. Thus our
assumption entails an initial path $b$ to $v$ so that $h_b > h_u1$. Let $u'$ be
$b_{\abs{u}}$: the node on the same level of $\G$ as $u$. Since $b$ is a path to $v$, it holds that
$(u', v) \in E$. Further, $h_b > h_u1$, it must be that $h_{u'} > h_{u}$.  By definition of
$E'$, the presence of $(u', v)$ where $h_{u'} > h_u$ precludes the edge $(u,v)$ from being in $E'$
--- a contradiction.
\end{proof}

Note that while it is possible for two nodes with different profiles to share a child in $\G$,
Lemma \ref{Gprime_Captures_Profiles} precludes this possibility in $\Gprime$. If two nodes join in
$\Gprime$, they must have the same profile and be in the same equivalence class.  We can thus
conflate nodes and equivalence classes, and for every edge $(u,v) \in E'$, consider $[v]$ to be the
child of $[u]$.  Lemma~\ref{Gprime_Captures_Profiles} then entails that the class $[u]$ can have at
most two children: the class of $F$-nodes with profile $h_u1$, and the class of non-$F$-nodes with
profile $h_u0$. We call the first class the $F$-child of $[u]$, and the second class the
non-$F$-child of $[u]$.

By using lexicographic ordering we can derive the preorder for each level $i\splus1$ of the run \DAG
solely from the preorder for the previous level $i$.  To determine the relation between two nodes,
we need only know the relation between the parents of those nodes, and whether the nodes are
$F$-nodes. Formally, we have the following.

\begin{lem}\label{Lexicographic_Ordering}
For all nodes $u, v$ on level $i$, and nodes $u', v'$ where $E'(u,u')$
and $E'(v,v')$\emph{:}
\begin{iteMize}{$\bullet$}
\item If $u \prec_i v$, then $u' \prec_{i+1} v'$.
\item If $u \approx_{i} v$ and either both $u'$ and $v'$ are $F$-nodes, or neither are $F$-nodes,
then $u' \approx_{i+1} v'$.
\item If $u \approx_i v$ and $v'$ is an $F$-node while $u'$ is not, then $u' \prec_{i+1} v'$.
\end{iteMize}
\end{lem}
\begin{proof}
If $u \prec_i v$, then $h_u < h_v$ and, by Lemma~\ref{Gprime_Captures_Profiles}, we know that $h_{u'} \in
\set{h_u0, h_u1}$ must be smaller than $h_{v'} \in \set{h_v0, h_v1}$, implying that $u' \prec_{i+1}
v'$.  If $u \approx_i v$, we have three sub-cases. If $v'$ is an $F$-node and
$u'$ is not, then $h_{u'}=h_u0=h_v0<h_v1=h_{v'}$ and $u' \prec_{i+1} v'$. If 
both $u'$ and $v'$ are $F$-nodes, then $h_{u'}=h_u1=h_v1=h_{v'}$ and $u' \approx_i
v'$.  Finally, if neither $u'$ nor $v'$ are $F$-nodes, then $h_{u'}=h_u0=h_v0=h_{v'}$ and $u'
\approx_i v'$.
\end{proof}

We now demonstrate that by keeping only edges associated with lexicographically maximal profiles,
$\Gprime$ 
captures an accepting path from $\G$.  

\begin{lem}\label{Lexicographic_Edge_Pruning}
$\Gprime$ has an accepting path iff $\G$ has an accepting path.
\end{lem}
\begin{proof}
In one direction, if $\Gprime$ has an accepting path, then its superset $\G$ has the same path.

In the other direction, assume $\G$ has an accepting path. Consider the set $P$ of accepting paths
in $\G$. We prove that there is a lexicographically maximal element $\pi \in P$.  To begin,
we construct an infinite sequence, $P_0,P_1,\ldots$, of subsets of $P$ such that the elements of
$P_i$ are lexicographically maximal in the first $i\splus1$ positions.  If $P$ contains paths starting in
an $F$-node, then $P_0 = \set{b\mid b \in P,~b_0\text{ is an $F$-node}}$ is all elements beginning
in $F$-nodes . Otherwise $P_0=P$. Inductively, if $P_i$ contains an element $b$ such that $b_{i+1}$
is an $F$-node, then $P_{i+1} = \set{b\mid b \in P_i,~b_{i+1}\text{ is an $F$-node}}$.  Otherwise
$P_{i+1}=P_i$. For convenience, define the predecessor of $P_i$ to be $P$ if $i=0$, and $P_{i-1}$
otherwise.  Note that since $\G$ has an accepting path, $P$ is non-empty. Further, every set $P_i$
is not equal to its predecessor $P'$ only when there is a path in $P'$ with an $F$-node in the $i$th
position. In this case, that path is in $P_i$. Thus every $P_i$ is non-empty. 

First, we prove that there is a path $\pi \in \bigcap_{i \geq 0} P_i$.  Consider the sequence
$U_0,U_1,U_2,\ldots$ where $U_i$ is the set of nodes that occur at position $i$ in runs in
$P_i$.  Formally, $U_i=\set{u \mid u \in \G,~b \in P_i,~u=b_i}$. Each node in $U_{i+1}$ has a parent in
$U_i$, although it may not have a child in $U_{i+2}$. We can thus connect the nodes in
$\bigcup_{i>0} U_i$ to their parents, forming a sub-\DAG of $\G$. As every $P_i$ is non-empty, every
$U_i$ is non-empty, and this \DAG has infinitely many nodes. Since each node has at most $n$
children, by \konig's Lemma there is an initial path $\pi$ through this \DAG, and thus through $\G$.
We now show by induction that $\pi \in P_i$ for every $i$. As a base case, $\pi \in P$. Inductively,
assume $\pi$ is in the predecessor $P'$ of $P_i$. The set $P_{i}$ is either $P'$, in which case $\pi
\in P_i$, or the set $\set{b\mid b \in P',~b_{i}\text{ is an $F$-node}}$.  In this latter case, as
$U_i$ consists only of $F$-nodes, the node $\pi_i$ must be an $F$-node.  and $\pi \in P_i$. 

Second, having established that there must be an element $\pi \in \bigcap_{i \geq 0} P_i$, we prove
$\pi$ is lexicographically maximal in $P$.  Assume by way of contradiction that there exists an
accepting path $\pi'$ so that $h_{\pi'} > h_\pi$. Let $k$ be the first point where $h_{\pi'}$
differs from $h_\pi$.  At this point, it must be that $\pi_k$ is not an $F$ node, while $\pi'_k$ is
an $F$ node.  However, $\pi'$ is an accepting path that shares a profile with $\pi$ up until this
point. As $\pi$ is in the predecessor $P'$ of $P_k$, it must also be that $\pi'$ is in $P'$.  By
definition, $P_k$ then would be $\set{b \mid b \in P',~ b_k\text{ is an $F$-node}}$. This would
imply $\pi \not \in P_k$, a contradiction.

Finally, we demonstrate that every edge in $\pi$ occurs in $\Gprime$. Assume by way of contradiction
that some edge $(\pi_i,\pi_{i+1})$ is in $E$ but not in $E'$. This implies there is a node $u$ on
level $i$ such that $(u,\pi_{i+1})$ is in $E$ and $\pi_i \prec_i u$. Since $u \in \G$, there is an
initial path $b$ to $u$.  Thus, the path $b,u,\pi_{i+1},\pi_{i+2}\ldots$ is an accepting path in
$\G$.  This path would be lexicographically larger than $\pi$, contradicting the second claim above.
Hence, we conclude $\pi$ is an accepting path in $\Gprime$.
\end{proof}

In the next stage, we remove from $\Gprime$ finite nodes. Let 
$\Gdubprime = \Gprime~\setminus \set{v\mid v\text { is finite in }\Gprime}$.  Note there may be
nodes that are not finite in $\G$, but are finite in $\Gprime$.  It is not
hard to see that $\G$ may have infinitely many $F$-nodes and still not contain a path with
infinitely many $F$-nodes.  Indeed, $\G$ may have infinitely many paths each with finitely many
$F$-nodes. We now show that the transition from $\G$ via $\Gprime$ to $\Gdubprime$ removes this
possibility, and the presence of infinitely many $F$-nodes in $\Gdubprime$ does imply the existence
of a path with infinitely many $F$-nodes.

\begin{lem}\label{Lexicographic_Node_Pruning}
$\G$ has an accepting path iff $\Gdubprime$ has infinitely many $F$-nodes.
\end{lem}
\begin{proof}
If $\G$ has an accepting path, then by Lemma~\ref{Lexicographic_Edge_Pruning} the \DAG $\Gprime$
contains an accepting path. Every node in this path is infinite in $\Gprime$, and thus this path is
preserved in $\Gdubprime$. This path contains infinitely many $F$-nodes, and thus $\Gdubprime$
contains infinitely many $F$-nodes.

In the other direction, we consider the \DAG over equivalence classes induced by $\Gdubprime$.
Given a node $u$ in $\Gdubprime$, recall that its equivalence class in $\Gdubprime$ contains all
states $v$ such that $v \in \Gdubprime$ and $h_u=h_v$. Given two equivalence classes $U$ and $V$,
recall that $V$ is a child of $U$ when there are $u \in U$,  $v \in V$, and  $E''(u,v)$. As
mentioned above, once we have pruned edges not in $\Gprime$, two nodes of different 
classes cannot join. Thus this \DAG is a tree.  Further, as every node $u$ in $\Gdubprime$ is
infinite and has a child, its equivalence class must also have a child.  Thus the \DAG of
classes in $\Gdubprime$ is a leafless tree. The width of this tree must monotonically
increase and is bounded by $n$. It follows that at some level $j$ the tree reaches a stable width. We call this
level $j$ the \emph{stabilization level} of $\G$.

After the stabilization level, each class $U$ has exactly one child: as noted above, $U$
cannot have zero children, and if $U$ had two children the width of the tree would increase.
Therefore, we identify each equivalence class on level $j$ of $\Gdubprime$ with its unique branch of
children in $\Gdubprime$, which we term its \emph{pipe}. These pipes form a
partition of nodes in $\Gdubprime$ after $j$. Every node in these pipes has an ancestor, or it would not
be in the \DAG, and has a child, or it would not be infinite and in $\Gdubprime$. Therefore each
node is part of an infinite path in this pipe.  Thus, the pipe with infinitely many
$F$-classes contains only accepting paths. These paths are accepting in $\G$, which
subsumes $\Gdubprime$.
\end{proof}

In the proof above we demonstrated there is a \emph{stabilization level} $j$ at which the number of
equivalence classes in $\Gdubprime$ stabilized, and discussed the {\em pipes} of $\Gdubprime$: the single
chain of descendants from each equivalence class on the stabilization level $j$ of $\Gdubprime$.

\begin{exa}
Figure~\ref{Fig:GDubPrime} displays $\Gdubprime$ for the example of Figure~\ref{Fig:Automaton}.
Edges removed from $\Gprime$ are dotted: at levels 1 and 3 where both $q$ and $r$ transition
to $r$.  When both $r$ and $s$ transition to $t$, they have the same profile and both edges remain.
The removed edges render all but the first $q$-node finite in $\Gprime$. The stabilization level is
$0$.
\end{exa}
\vspace{-0.0in}

\begin{figure}
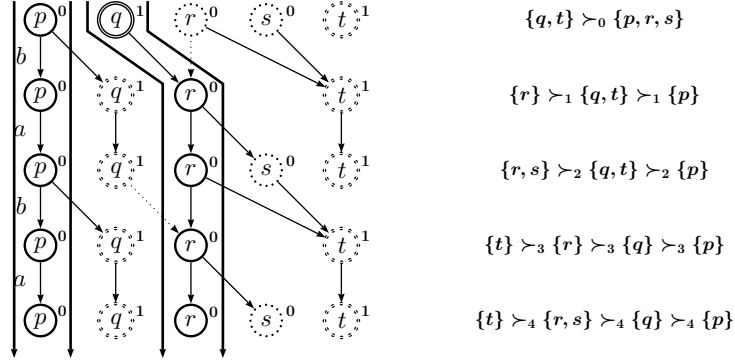

\centering
{
\begin{postscript}
\SmallPicture\VCDraw{\begin{VCPicture}{(-4,-12)(15,2)}
\def\level{0}\def\offset{0}
\State[p]{(0,\offset)}{p\level} \FinalState[q]{(2,\offset)}{q\level}
\ChgStateLineStyle{dotted}
\State[r]{(4,\offset)}{r\level} \State[s]{(6,\offset)}{s\level} \FinalState[t]{(8,\offset)}{t\level}
\RstStateLineStyle
\def\level{1} \def\prevlevel{0} \def\offset{-2}
\State[p]{(0,\offset)}{p\level}\State[r]{(4,\offset)}{r\level}
\ChgStateLineStyle{dotted}
\FinalState[q]{(2,\offset)}{q\level} \FinalState[t]{(8,\offset)}{t\level}
\RstStateLineStyle
\EdgeR{p\prevlevel}{p\level}{b~~~~~~~} \EdgeL{p\prevlevel}{q\level}{} 
\ChgEdgeLineStyle{dotted}\EdgeL{r\prevlevel}{r\level}{}\RstEdgeLineStyle
\EdgeL{q\prevlevel}{r\level}{} \EdgeL{r\prevlevel}{t\level}{} \EdgeL{s\prevlevel}{t\level}{}
\def\prevlevel{1} \def\level{2}\def\offset{-4}
\State[r]{(4,\offset)}{r\level} \State[p]{(0,\offset)}{p\level}
\ChgStateLineStyle{dotted}
\FinalState[q]{(2,\offset)}{q\level} \State[s]{(6,\offset)}{s\level} \FinalState[t]{(8,\offset)}{t\level}
\RstStateLineStyle
\EdgeR{p\prevlevel}{p\level}{a~~~~~~~} \EdgeL{r\prevlevel}{r\level}{} \EdgeL{q\prevlevel}{q\level}{}
\EdgeL{t\prevlevel}{t\level}{} \EdgeL{r\prevlevel}{s\level}{} 
\def\prevlevel{2} \def\level{3}\def\offset{-6}
\State[p]{(0,\offset)}{p\level} \State[r]{(4,\offset)}{r\level}
\ChgStateLineStyle{dotted}
\FinalState[t]{(8,\offset)}{t\level} \FinalState[q]{(2,\offset)}{q\level} 
\RstStateLineStyle
\EdgeR{p\prevlevel}{p\level}{b~~~~~~~} \EdgeL{r\prevlevel}{r\level}{} \EdgeL{p\prevlevel}{q\level}{}
\ChgEdgeLineStyle{dotted}\EdgeL{q\prevlevel}{r\level}{}\RstEdgeLineStyle
\EdgeL{r\prevlevel}{t\level}{} \EdgeL{s\prevlevel}{t\level}{}
\def\prevlevel{3} \def\level{4}\def\offset{-8}
\State[r]{(4,\offset)}{r\level} \State[p]{(0,\offset)}{p\level}
\ChgStateLineStyle{dotted}
\FinalState[q]{(2,\offset)}{q\level} \State[s]{(6,\offset)}{s\level} \FinalState[t]{(8,\offset)}{t\level}
\RstStateLineStyle
\EdgeR{p\prevlevel}{p\level}{a~~~~~~~} \EdgeL{r\prevlevel}{r\level}{} \EdgeL{q\prevlevel}{q\level}{}
\EdgeL{t\prevlevel}{t\level}{} \EdgeL{r\prevlevel}{s\level}{} 
\hide{
\def\prevlevel{4} \def\level{5}\def\offset{-10}
\State[r]{(4,\offset)}{r\level} \State[p]{(0,\offset)}{p\level}
\ChgStateLineStyle{dotted}
\FinalState[q]{(2,\offset)}{q\level} \State[s]{(6,\offset)}{s\level} \FinalState[t]{(8,\offset)}{t\level}
\RstStateLineStyle
\EdgeR{p\prevlevel}{p\level}{a~~~~~~~} \EdgeL{r\prevlevel}{r\level}{} \EdgeL{q\prevlevel}{q\level}{}
\EdgeL{s\prevlevel}{s\level}{} \EdgeL{t\prevlevel}{t\level}{} \EdgeL{r\prevlevel}{s\level}{} 
\def\prevlevel{5} \def\level{6}\def\offset{-12}
\State[p]{(0,\offset)}{p\level} \State[r]{(4,\offset)}{r\level}
\ChgStateLineStyle{dotted}
\FinalState[t]{(8,\offset)}{t\level} \FinalState[q]{(2,\offset)}{q\level} 
\RstStateLineStyle
\EdgeR{p\prevlevel}{p\level}{b~~~~~~~} \EdgeL{r\prevlevel}{r\level}{} \EdgeL{p\prevlevel}{q\level}{}
 \EdgeL{r\prevlevel}{t\level}{} \EdgeL{s\prevlevel}{t\level}{}
\ChgEdgeLineStyle{dotted}\EdgeL{q\prevlevel}{r\level}{}\RstEdgeLineStyle
}

\psline[linewidth=2pt] (-0.7,0.5)(-0.7,-9)
\psline[linewidth=2pt] (0.8,0.5)(0.8,-9)
\psline[linewidth=2pt] (1.25,0.5)(1.25,-0.3)(3.25,-1.7)(3.25,-9)
\psline[linewidth=2pt] (2.85,0.5)(2.85,-0.3)(4.85,-1.7)(4.85,-9)
\begin{boldmath}

\multirput(0.58,0.25)(0.0,-2.0){5}{\large$0$}
\multirput(2.65,0.25)(0.0,-2.0){5}{\large$1$}
\multirput(4.60,0.22)(0.0,-2.0){5}{\large$0$}
\multirput(6.65,0.25)(0.0,-4.0){2}{\large$0$}
\multirput(6.65,-7.75)(0.0,-2.0){1}{\large$0$}
\multirput(8.65,0.25)(0.0,-2.0){5}{\large$1$}

\rput[u](15,0){\Large $\set{q,t}\succ_0\set{p,r,s}$}
\rput[u](15,-2){\Large $\set{r}\succ_1\set{q,t}\succ_1\set{p}$}
\rput[u](15,-4){\Large $\set{r,s}\succ_2\set{q,t}\succ_2\set{p}$}
\rput[u](15,-6){\Large $\set{t}\succ_3\set{r}\succ_3\set{q}\succ_3\set{p}$}
\rput[u](15,-8){\Large $\set{t}\succ_4\set{r,s}\succ_4\set{q}\succ_4\set{p}$}
\end{boldmath}

\end{VCPicture}}
\end{postscript}\qquad\qquad\qquad\qquad
}
\vspace{-0.5in}
\caption{The run \DAG \Gdubprime, where dotted edges were removed from \G and dotted states were removed from
\Gprime. Nodes are superscripted with their $\Lambda$-labels. Bold lines denote the pipes of \Gdubprime. The
lexicographic order of equivalence classes for each level of \Gprime is to the right.
}\label{Fig:GDubPrime}
\end{figure}

\subsection{Complementing With Profiles}
We now complement $\A$ by constructing an NBW, $\A_S$, that employs
Lemma~\ref{Lexicographic_Node_Pruning} to determine if a word is in $L(\A)$. This construction is a
reformulation of the slice-based approach of \cite{KW08} in the framework of run \DAGs\forappendix{:
see Appendix~\ref{App:Slices}}{}.  The NBW $\A_S$ tracks the levels of $\Gprime$ and guesses which
nodes are finite in $\Gprime$ and therefore do not occur in $\Gdubprime$. To track $\Gprime$, the
automaton $\A_S$ stores at each point in time a set $S$ of states that occurs on each level. The
sets $S$ are labeled with a guess of which nodes are finite and which are infinite.  States that are
guessed to be infinite, and thus correspond to nodes in $\Gdubprime$, are labeled $\top$, and states
that are guessed to be finite, and thus omitted from $\Gdubprime$, are labeled $\bot$.  In order to
track the edges of $\Gprime$, and thus maintain this labeling, $\A_S$ needs to know the
lexicographic order of nodes.  Thus $\A_S$ also maintains the preorder $\preceq_i$ over states on
the corresponding level of the run \DAG.  To enforce that states labeled $\bot$ are indeed finite,
$\A_S$ employs the cut-point construction of Miyano and Hayashi \cite{MH84}, keeping an ``obligation
set'' of states currently being verified as finite.  Finally, to ensure the word is rejected, $\A_S$
must enforce that there are finitely many $F$-nodes in $\Gdubprime$.  To do so, $S_A$ uses a bit $b$
to guess the level from which no more $F$-nodes appear in $\Gdubprime$. After this point, 
$F$-nodes must be labeled $\bot$.

Before we define $\A_S$, we formalize {\em preordered subsets} and 
operations over them.  For a set $Q$ of states, define ${\PSQ} = \set{\zug{S,\preceq}\mid S \subseteq
Q \text{ and }\preceq\text{ is a preorder over $S$}}$ to be the set of preordered subsets of $Q$.
Let $\zug{S,\preceq}$ be an element in ${\PSQ}$.  When considering the successors of
a state, we want to consider edges that remain in $\Gprime$.
For every state $q \in S$ and $\sigma \in \Sigma$, define $\rho_\zug{S,\preceq}(q,\sigma) = \{r \in
\rho(q,\sigma)\mid \text{for every $q' \in S$, if $r \in \rho(q',\sigma)$ then }q' \preceq q\}$.  
Now define the {\em $\sigma$-successor of $\zug{S,\preceq}$} as the tuple
$\zug{\rho(S,\sigma),\preceq'}$, where for every $q, r \in S$,  $q' \in \rho_\zug{S,\preceq}(q, \sigma)$, and
$r' \in \rho_\zug{S,\preceq}(r, \sigma)$:
\begin{iteMize}{$\bullet$}
\item If $q \prec r$, then $q' \prec' r'$
\item If $q \approx r$ and either both $r' \in F$ and $q' \in F$, or both $r' \not\in F$ and $q'
\not \in F$, then $q' \approx' r'$.
\item If $q \approx r$ and one of $q'$ and $r'$, say $r'$, is in $F$ while the other, $q'$, is not, then $q' \prec' r'$. 
\end{iteMize}

We now define $\A_S$. 
The states of $\A_S$ are tuples
$\zug{S,\preceq,\lambda,O,b}$ where: $\zug{S,\preceq} \in {\PSQ}$ is preordered subset of $Q$;~
$\lambda \colon S \to \set{\top,\bot}$ is a labeling 
indicating which states are guessed to be finite ($\bot$) or
infinite ($\top$);~ $O \subseteq S$ is the
obligation set;~ and $b \in \set{0,1}$ is a bit indicating whether we have seen the last $F$-node in
$\Gdubprime$. 
To transition between states of $\A_s$, say that 
\emph{$\textbf{t}'=\zug{S',\preceq',\lambda',O',b'}$ follows
$\textbf{t}=\zug{S,\preceq,\lambda,O,b}$ under $\sigma$} when: 
\begin{enumerate}[(1)]
\item $\zug{S',\preceq'}$ is the $\sigma$-successor of $\zug{S,\preceq}$. 
\item $\lambda'$ is such that for every $q \in S$: 
\begin{iteMize}{$\bullet$}
\item If $\lambda(q)=\top$, then there exists $r \in \rho_\zug{S,\preceq}(q,\sigma)$ such that $\lambda'(r)=\top$,
\item If $\lambda(q)=\bot$, then for every $r \in \rho_\zug{S,\preceq}(q,\sigma)$, it holds that $\lambda'(r)=\bot$.
\end{iteMize}
\item $O' = \begin{cases} 
   {\bigcup_{q \in O}}~\rho_\zug{S,\preceq}(q,\sigma) & O \neq \emptyset,\\
   \{q\mid q \in S' \text{ and }~\lambda'(q)=\bot\} & O = \emptyset.\\
\end{cases}$
\item $b' \geq b$.
\end{enumerate}
We want to ensure that runs of $A_S$ reach a suffix where all $F$-nodes are finite. To this
end, given a state of $\A_S$ $\zug{S,\preceq,\lambda,O,b}$, we say that $\lambda$ is \emph{$F$-free} if
for every $q \in S \cap F$ we have $\lambda(q)=\bot$.

\begin{defi}\label{Slice_Def}
For an NBW $\A = \zug{\Sigma, Q, Q^{in}, \rho, F}$, let $\A_S$ be the NBW 
$\zug{\Sigma, Q_{S}, Q^{in}_{S}, \rho_{S}, F_{S}}$, where:
\begin{iteMize}{$\bullet$}
\item $Q_S = \set{\zug{S,\preceq,\lambda,O,b}\mid \text{if $b=1$ then $\lambda$ is $F$-free}}$,
\item $Q^{in}_S = \set{\zug{Q^{in},\preceq,\lambda,\emptyset,0}\mid \text{for all $q,r \in Q^{in}$, }
q \preceq r\text{ iff } q \not \in F \text{ or } r\in F}$,
\item $\rho_S(\textbf{t},\sigma) = \set{\textbf{t}'\mid \textbf{t}'\text{ follows $\textbf{t}$ under $\sigma$}}$, and
\item $F_S = \set{\zug{S,\preceq,\lambda,\emptyset,1}}$.
\end{iteMize}
\end{defi}

\cbstart
We divide runs of $\A_S$ into two parts. The {\em prefix} of a run is the initial sequence of
states in which $b_i$ is $0$, and the {\em suffix} is the remaining sequence states, in
which $b_i$ is $1$. A run without a suffix, where $b$ stays $0$ for the entire run, has no accepting
states.
\cbend

\begin{theorem}\label{Slice_Complement}
For every NBW $\A$, it holds that $L(\A_S)=\overline{L(\A)}$.
\end{theorem}
\begin{proof}
\cbstart
Consider a word $w \in \Sigma^\omega$ and the run \DAG $\G$. We first make the following claims
about every infinite run ${\bf t}_0,{\bf t}_1,\ldots$, where
$\textbf{t}_i=\zug{S_i,\preceq_i,\lambda_i,O_i,b_i}$.  For convenience, define
$\mcS_i=\zug{S_i,\preceq_i}$. 
\begin{enumerate}[(1)]

\item\label{Step:S} \textit{The states in $S_i$ are precisely $\set{q\mid \rzug{q,i} \in \G}$.}\\
We exploit this claim to conflate a state $q$ in the $i$th state with the node $\rzug{q,i}$,
and speak of states in $S_i$ being in, being finite in, and being infinite in a graph $\G$. 

\item\label{Step:Preorder} \textit{The preorder $\preceq_i$ is the projection of $\preceq$
onto states occurring at level $i$.} \\This follows from Lemma~\ref{Lexicographic_Ordering}
and the definition of one state in $\A_S$ following another.

\item\label{Step:Transitions} \textit{For every $p \in S_i$, $q \in S_{i+1}$, it holds that $q \in
\rho_{\mcS_i}(p,\sigma_i)$ iff $E'(\rzug{p,i},\rzug{q,i\splus1}).$}\\
This follows from the definitions of $E'$ and $\rho_\mcS$. 

\item\label{Step:O} \textit{$O_i$ is empty for infinitely many indices
  $i$ iff every state labeled $\bot$ is not in $\Gdubprime$}.\\ This
  follows from the cut-point construction of Miyano and
  Hayashi. \cite{MH84}.

\item\label{Step:Inf} \textit{Every state labeled $\top$ is in $\Gdubprime$}.\\
This follows from the definition of transitions between states: every $\top$-labeled state must
have a $\top$-labeled child, and thus is infinite in $\Gprime$ and in $\Gdubprime$.
\end{enumerate}

\noindent We can now prove the theorem. In one direction, assume there is an accepting run ${\bf
t}_0,{\bf t}_1,\ldots$. As this run is accepting, infinitely often $O_i=\emptyset$. By \ref{Step:O} and
\ref{Step:Inf}, this implies the states in $S_i$ are correctly labeled $\top$ when and only when they occur
in $\Gdubprime$.  Further, for this run to be accepting we must be able to divide the run into a prefix,
and suffix as specified above. In the suffix no state in $F$ can be
labeled $\top$, and thus no $F$-nodes occur in $\Gdubprime$ past this point. As only finitely many
$F$-nodes can occur before this point, by Lemma~\ref{Lexicographic_Edge_Pruning} $\G$ does not have
an accepting path and $w \not \in L(\A)$.

In the other direction, assume $w \not \in L(\A)$. This implies there are finitely many $F$-nodes in
$\Gdubprime$, and thus a level $j$ where the last $F$-node occurs.  We construct an accepting run
${\bf t}_0,{\bf t}_1,\ldots$, demonstrating along the way that we satisfy the requirements for
$\textbf{t}_{i+1}$ to be in $\rho_S(\textbf{t}_i,\sigma_i)$. Given $w$, the sequence
$\zug{S_0,\preceq_0},\zug{S_1,\preceq_1},\ldots$ of preordered subsets  is uniquely defined by
$\rho_S$.  There are many possible labelings $\lambda$. For every $i$, select $\lambda_i$ so that a
state $q \in S_i$ is labeled with $\top$ when $\rzug{q,i} \in \Gdubprime$, and $\bot$ when it is
not. Since every node in $\Gdubprime$ has a child, by \ref{Step:Transitions}, for every $p \in S_i$
where $\lambda_i(p)=\top$, there exist a $q \in \rho_{\mcS_i}(p,\sigma_i)$ so that
$\lambda_{i+1}(q)=\top$.  Further, every node labeled $\bot$ has only finitely many descendants, and
so for every $p \in S_i$ where $\lambda_i(p)=\bot$ and $q \in \rho_{\mcS_i}(p,\sigma_i)$, it holds
that $\lambda_{i+1}(q)=\bot$.  Therefore the transition from $\lambda_i$ to $\lambda_{i+1}$ satisfies
the requirements of $\rho_S$. The set $O_0=\emptyset$, and given the sets $S_i$ and labelings
$\lambda_i$, the sets $O_{i+1}$, $i \geq 0$ are again uniquely defined by $\rho_S$.  Finally, we
choose $b_i=0$ when $i<j$, and $b_i=1$ for $i \geq j$.  Since there are no $F$-nodes in $\Gdubprime$
past $j$, no $F$-node will be labeled $\top$ and all states past $j$ will be $F$-free. We have
satisfied the last requirement for the transitions from every ${\bf t}_i$ to ${\bf t}_{i+1}$ to be
valid, rendering this sequence a run. By \ref{Step:O}, infinitely often $O_i=\emptyset$, including
infinitely often after $j$, thus there are infinitely many states ${\bf t}_i$ where $b_i=1$ and
$O_i=\emptyset$, and this run is accepting.
\cbend
\end{proof} 

If $n\!=\!\abs{Q}$, the number of preordered subsets is roughly $(0.53n)^n$ \cite{Var80}. As there
are $2^n$ labelings, and a further $2^n$ obligation sets, the state space of $\A_s$ is at most
$(2n)^n$.  The slice-based automaton obtained in \cite{KW08} coincides with $\A_S$, modulo the
details of labeling states and the cut-point construction (see
Appendix~\ref{App:Slices}). Whereas the correctness proof in \cite{KW08} is given by means of
reduced split trees, here we proceed directly on the run \DAG.

\section{Retrospection}\label{Slices_To_Ranks}

Consider an NBW $\A$. So far, we presented two complementation constructions for $\A$, generating
the NBWs $\A^m_R$ and $\A_S$. In this section we present a third
construction, generating an NBW that combines the benefits of the two constructions above.  Both
constructions refer to the run \DAG of $\A$.  In the rank-based approach applied in $\A^m_R$, the
ranks assigned to a node bound the visits in accepting states yet to come. Thus, the ranks refer to
the future, making $\A^m_R$ inherently nondeterministic.  On the other hand, the NBW $\A_S$  refers
to both the past, using profiles to prune edges from $\G$, as well as to the future, by keeping in
$\Gdubprime$ only nodes that are infinite in $\Gprime$. Guessing which nodes are infinite and
labeling them $\top$ inherently introduces nondeterminism into the automaton.

Our first goal in the combined construction is to reduce this latter nondeterminism.  Recall that a
labeling is $F$-free if all the states in $F$ are labeled $\bot$.  Observe that the fewer labels of
$\bot$ (finite nodes) we have, the more difficult it is for a labeling to be $F$-free and,
consequently, the more difficult it is for a run of $\A_S$ to proceed to the $F$-free suffix in
which $b=1$.  It is therefore safe for $\A_S$ to underestimate which nodes to label $\bot$, as long
as the requirement to reach an $F$-free suffix is maintained.  We use this observation in order to
introduce a purely retrospective construction.

For a run \DAG $\G$, say that a level $k$ is an \emph{$F$-finite level} of $\G$ when all $F$-nodes
{\em after} level $k$ (i.e. on a level $k'$ where $k' > k$) are finite in $\Gprime$.  
By Lemma~\ref{Lexicographic_Node_Pruning}, $\G$ is rejecting iff there is a level after which
$\Gdubprime$ has no $F$-nodes. As finite nodes in $\Gprime$ are removed from $\Gdubprime$, we
have:

\begin{corollary}\label{Rejecting_iff_k}
A run \DAG $\G$ is rejecting iff it has an $F$-finite level.
\end{corollary}

\subsection{Retrospective Labeling}
The labeling function $\lambda$ used in the construction of $\A_S$ labels nodes by $\{\top,\bot\}$,
with $\bot$ standing for ``finite'' and $\top$ standing for ``infinite''. In this section we
introduce a variant of $\lambda$ that again maps nodes to $\{\top,\bot\}$ except that now $\top$
stands for ``unrestricted'', allowing us to underestimate which nodes to label
$\bot$. To capture the relaxed requirements on labelings, say that a labeling $\lambda$ is
\emph{legal} when every $\bot$-labeled node is finite in $\Gprime$.  This enables the automaton to
track the labeling and its effect on $F$-nodes only after it guesses that an $F$-finite level $k$
has been reached: all nodes {\em at or before} level $k$ (i.e. on a level $k'$ where $k' \leq k$) are
unrestricted, whereas $F$-nodes after level $k$ and their descendants are required to be finite. The only nondeterminism
in the automaton lies in guessing when the $F$-finite level has been reached.  This reduces the
branching degree of the automaton to 2, and renders it deterministic in the limit. 

The suggested new labeling is parametrized by the $F$-finite level $k$.  The labeling $\lambda^k$
is defined inductively over the levels of $\G$.  Let $S_i$ be the set of nodes on level $i$ of $\G$.
For $i \geq 0$, the function $\lambda^k \colon S_i \to \set{\top,\bot}$ is defined as follows:
\begin{iteMize}{$\bullet$}
\item If $i \leq k$, then for every $u \in S_i$ we define $\lambda^k(u)=\top$.
\item If $i > k$, then for every $u \in S_i$:
\begin{iteMize}{$-$}
\item If $u$ is an $F$-node, then $\lambda^k(u)=\bot$.
\item Otherwise, $\lambda^k(u)=\lambda^k(v)$, for a node $v$ where $E'(v,u)$. 
\end{iteMize}
\end{iteMize}
For $\lambda^k$ to be well defined when $i > k$ and $u$ is not an $F$-node, we need to show that
$\lambda^k(u)$ does not depend on the choice of the node $v$ where $E'(v,u)$ holds. By
Lemma~\ref{Gprime_Captures_Profiles}, all parents of a node in $\Gprime$ belong to the same
equivalence class. Therefore, it suffices to prove that all nodes in the same class share a label:
for all nodes $u$ and $u'$, if $u' \approx_{\abs{u}} u$ then $\lambda^k(u)=\lambda^k(u')$.  The
proof proceeds by an induction on $i=|u|$. 
Consider two nodes $u$ and $u'$ on level $i$ where $u' \approx_i u$. As a base case, if $i \leq k$, then $u$ and
$u'$ are labeled $\top$.  For $i > k$, if $u$ is an $F$-node, then $u'$ is also an $F$-node and
$\lambda^k(u)=\lambda^k(u')=\bot$.  Finally, if $u$ and $u'$ are both non-$F$-nodes, recall that all
parents of $u$ are in the same equivalence class $V$. As $u \approx_i u'$,
Lemma~\ref{Gprime_Captures_Profiles} implies that all parents of $u'$ are also in $V$. By the
induction hypothesis, all nodes in $V$ share a label, and thus $\lambda^k(u)=\lambda^k(u')$.

\begin{lem}\label{Slices_Make_Sense}
For a run \DAG $\G$ and $k \in \N$, the labeling $\lambda^k$ is legal iff $k$ is an $F$-finite
level for $\G$.
\end{lem}
\begin{proof}
If $\lambda^k$ is legal, then every $\bot$-labeled node is finite in $\Gprime$.  Every $F$-node
after level $k$ (i.e. on a level $i$ where $i > k$) is labeled $\bot$, and thus $k$ is an $F$-finite level for $\G$.  If $\lambda^k$ is
not legal, then there is a $\bot$-labeled node $u$ that is infinite in $\Gprime$. Every
ancestor of $u$ is also infinite.  Let $u'$ be the earliest ancestor of $u$ 
(possibly $u$ itself) so
that $\lambda^k(u')=\bot$. Observe that only nodes after level $k$ can be $\bot$-labeled, and so
$u'$ is on a level $i > k$. It must be that $u'$ is an $F$-node: otherwise it would inherit its
parents' label, and by assumption the parents of $u'$ are $\top$-labeled. Thus, $u'$ is an
$F$-node after level $k$ that is infinite in $\Gprime$, and $k$ is not an $F$-finite level for
$\G$.
\end{proof}

\begin{corollary}\label{Slices_Complement}
A run \DAG $\G$ is rejecting iff, for some $k$, the labeling $\lambda^k$ is legal.
\end{corollary}

\subsection{From Labelings to Rankings}
In this section we derive an odd ranking for $\G$ from the function $\lambda^k$, thus unifying the
retrospective analysis behind $\lambda^k$ with the rank-based analysis of \cite{KV01c}.  Consider
again the \DAG $\Gprime$ and 
the
function $\lambda^k$.  Recall that every equivalence class $U$ has at
most two child equivalence classes, one $F$-class and one non-$F$-class.  After the $F$-finite level
$k$, only non-$F$-classes can be labeled $\top$. Hence, after level $k$, every $\top$-labeled
equivalence class $U$ can only have a one child that is $\top$-labeled.  For every 
class $U$ on level $k$, we consider this possibly infinite sequence of $\top$-labeled
non-$F$-children. The odd ranking we are going to define, termed the {\em retrospective ranking},
gives these sequences of $\top$-labeled children odd ranks. The $\bot$-labeled classes, which
lie between these sequences of $\top$-labeled classes, are assigned even ranks. The ranks increase
in inverse lexicographic order, i.e. the maximal $\top$-labeled class in a level is given rank 1. As
with $\lambda^k$, the retrospective ranking is parametrized by $k$. The primary insight that allows
this ranking is that there is no need to distinguish between two adjacent $\bot$-labeled classes.
Formally, we have the following.

\begin{defi}[$k$-retrospective ranking]\label{kret_ranks}
Consider a run \DAG $\G$, $k \in \N$, and the labeling $\lambda^k\colon \G \to \{\top,\bot\}$.  Let
$m=2\abs{Q\setminus F}$. For a node $u$ on level $i$ of $\G$, let $\prerank(u)$ be the number of
$\top$-labeled classes lexicographically larger than $u$; $\prerank(u) = \abs{\set{[v]\mid \lambda^k(v)=\top \text{
and } u \prec_i v}}$.  The \emph{$k$-retrospective ranking\/} of $\Gprime$ is the function ${\bf
r}^k \colon V \to \set{0..m}$ defined for every node $u$ on level $i$ as follows.
$${\bf r}^k(u) = \begin{cases}
m & \text{if $i \leq k$,}\\
2\prerank(u) & \text{if $i > k$ and $\lambda^k(u)=\bot$,}\\
2\prerank(u) + 1 & \text{if $i > k$ and $\lambda^k(u)=\top$.} 
\end{cases}$$
\end{defi}

Note that ${\bf r}^k$ is tight. As defined in Section \ref{Sect:Tight}, a ranking is
tight if there exists an $i \in \N$ such that, for every level $l \geq i$, all odd ranks below ${\it
max\_rank}({\bf r},l)$ appear on level $l$. For ${\bf r}^k$ this level is $k+1$, after which each
$\top$-labeled class is given the odd rank greater by two than the rank of the next lexicographically larger
$\top$-labeled class.

\begin{lem}\label{Ranking_Respects_Preorder} For every $k \in \emph{\N}$, the following hold:
\begin{enumerate}[\em(1)]
\item If $u \prec_{\abs{u}} u'$ then ${\bf r}^k(u) \geq {\bf r}^k(u')$.
\item If $(u,v) \in E'$, then ${\bf r}^k(u) \geq {\bf r}^k(v)$.
\end{enumerate}
\end{lem}
\begin{proof}
As both claims are trivial when $u$ is at or before level $k$, assume $u$ is on level $i > k$. To prove
the first claim, note that $\prerank(u) \geq \prerank(u')$: every class, $\top$-labeled
or not, that is larger than $u'$ must also be larger than $u$. If $\prerank(u) > \prerank(u')$, then
(1) follows immediately.  Otherwise $\prerank(u) = \prerank(u')$, which implies that
$\lambda^k(u')=\bot$: otherwise $[u']$ would be a $\top$-labeled equivalence class larger than $u$,
but not larger than itself. Thus ${\bf r}^k(u')=2\prerank(u)$, and ${\bf r}^k(u) \in
\set{2\prerank(u),~2\prerank(u)\!\splus 1}$ is at least ${\bf r}^k(u')$.

As a step towards proving the second claim, we show that $\prerank(u) \geq \prerank(v)$. Consider
every $\top$-labeled class $[v']$ where $v \prec_{i+1} v'$. The class $[v']$ must have a
$\top$-labeled parent $[u']$.  Since $v \prec_{i+1} v'$, the contrapositive of
Lemma \ref{Lexicographic_Ordering}, part 1, entails that $u \preceq_i u'$.  By the definition of $\lambda^k$,
the class
$[u']$ can only have one $\top$-labeled child class: $[v']$.  We have thus established
that for every $\top$-labeled class larger than $v$, there is a unique $\top$-labeled class larger
than $u$, and can conclude that $\prerank(u) \geq \prerank(v)$.  We now show by contradiction that
${\bf r}^k(u) \geq {\bf r}^k(v)$. For ${\bf r}^k(u) < {\bf r}^k(v)$, it must be that
$\prerank(u) = \prerank(v)$, that ${\bf r}^k(u)$ = $2\prerank(u)$, and that ${\bf
r}^k(v)=2\prerank(u)\splus1$. In this case, $\lambda^k(u)=\bot$ and $\lambda^k(v)=\top$.  Since a
$\bot$-labeled node cannot have a $\top$-labeled child in $\Gprime$, this is impossible.
\end{proof}

When $k$ is an $F$-finite level of $\G$, the $k$-retrospective ranking is an $m$-bounded odd ranking.

\begin{lem}\label{Slices_Provide_Ranking}
For a run \DAG $\G$ and $k \in \emph{\N}$, the function ${\bf r}^k$ is a ranking bounded by $m$.
Further, if the labeling $\lambda^k$ is legal then ${\bf r}^k$ is an odd ranking.
\end{lem}
\begin{proof}
There are three requirements for ${\bf r}^k$ to be a ranking bounded by $m$:
\begin{enumerate}[(1)]
\item \textit{Every $F$-node must have an even rank.} At or before 
level $k$,
every node has rank $m$,
which is even.  After $k$ only $\top$-labeled nodes are given odd ranks, and every $F$-node is
labeled $\bot$.

\item \textit{For every $(u,v) \in E$, it must hold that ${\bf r}^k(u) \geq {\bf r}^k(v)$.} 
If $u$ is at or before level $k$, then it has the maximal rank of $m$. If $u$ is after level $k$, we
consider two cases: edges in $E'$, and edges in $E \setminus E'$.  For edges in $E'$, this follows
from Lemma \ref{Ranking_Respects_Preorder} (2). For edges $(u,v) \in E \setminus E'$, we know there
exists a $u'$ where $u \prec_{\abs{u}} u'$ and $(u',v) \in E'$. By Lemma
\ref{Ranking_Respects_Preorder}, ${\bf r}^k(u) \geq {\bf r}^k(u') \geq {\bf r}^k(v)$.

\item \textit{The rank is bounded by $m$.} No $F$-node can be $\top$-labeled. Thus the maximum
number of $\top$-labeled classes on every level is $\abs{Q\setminus F}$.  The largest possible
rank is given to a node smaller than all $\top$-labeled classes, which must be
be a $F$-node and $\bot$-labeled.  Thus, this node is given a rank of at most $m=2\abs{Q\setminus
F}$.\smallskip
\end{enumerate}

\noindent It remains to show that if $\lambda^k$ is legal, then ${\bf r}^k$ is an odd ranking.  Consider an
infinite path $u_0,u_1,\ldots$ in $\G$.  We demonstrate that for every $i > k$ 
such that ${\bf r}^k(u_i)$ is an even rank $e$, there exists $i' > i$ such that ${\bf r}^k(u_{i'}) \neq e$. 
Since
a path cannot increase in rank, this implies ${\bf r}^k(u_{i'}) < e$. To do so, define the
sequence $U_i,U_{i+1},\ldots$, of sets of nodes inductively as follows.  Let $U_i = \set{v\mid {\bf
r}^k(v)=e\hide{,~u \preceq_i v}}$.  For every $j \geq i$, let $U_{j+1} = \set{v \mid v' \in
U_j,~(v',v) \in E'}$.  As ${\bf r}^k(v)$ is even only when $\lambda^k(v)=\bot$, if $\lambda^k$ is
legal then every node given an even rank (such as $e$) must be finite in $\Gprime$. Therefore every
element of $U_i$ is finite in $\Gprime$, and thus at some $i' > i$, the set $U_{i'}$ is empty. Since
$U_{i'}$ is empty, to establish that ${\bf r}^k(u_{i'}) \neq e$, it is sufficient to prove  that for
every $j$, if ${\bf r}^k(u_{j})=e$, then $u_{j} \in U_{j}$.

To show that ${\bf r}^k(u_{j})=e$ entails $u_{j} \in U_{j}$, we prove a stronger claim: for every $j
\geq i$ and $v$ on level $j$, if $u_j \preceq_j v$ and ${\bf r}^k(v)=e$, then $v \in U_j$.  We
proceed by induction over $j$.  For the base case of $j=i$, this follows from the definition
of $U_i$. For the inductive step, take a 
node
$v$ on level $j\splus1$ where ${\bf r}^k(v)=e$ and $u_{j+1}
\preceq_{j+1} v$.  We consider two cases. If ${\bf r}^k(u_{j+1}) \neq e$ then the path 
from $u_i$ to $u_{j+1}$ entails that ${\bf r}^k(u_{j+1}) < e$, and this case of the subclaim
follows from Lemma \ref{Ranking_Respects_Preorder} (1).  Otherwise, it holds that ${\bf
r}^k(u_{j+1}) = e$, and thus ${\bf r}^k(u_j)=e$. Let $u'$ and $v'$ be nodes on level $j$ so that
$(u',u_{j+1}) \in E'$ and $(v',v) \in E'$.  As $u_{j+1} \preceq_{j+1} v$, the contrapositive of Lemma
\ref{Lexicographic_Ordering}, part 1, entails that $u' \preceq_j v'$. Further, since $(u',u_{j+1})
\in E'$ and $(u_j, u_{j+1}) \in E$, we know $u_j \preceq_j u'$.  By transitivity we can thus
conclude that $u_j \preceq_j v'$, which along with Lemma \ref{Ranking_Respects_Preorder} (1) entails
${\bf r}^k(u')=e \geq {\bf r}^k(v')$. As $(v', v) \in E$, Lemma \ref{Ranking_Respects_Preorder} (2)
entails that ${\bf r}^k(v') \geq {\bf r}^k(v)=e$. Thus ${\bf r}^k(v')=e$, and by the inductive
hypothesis $v' \in U_j$. As $E'(v',v)$ holds, by definition $v \in U_{j+1}$, and our subclaim is proven.
\end{proof}


The ranking of Definition~\ref{kret_ranks} is termed {\em retrospective} as it relies on the relative
lexicographic order of equivalence classes; this order is determined purely by the history of
nodes in the run \DAG, not by looking forward to see which descendants are infinite or $F$-free in
some subgraph of $\G$.

\begin{exa}
Figure~\ref{Fig:GPrime} displays $\lambda^0$ and the 0-retrospective ranking of our running example.
In the prospective ranking (Figure~\ref{Fig:GDubPrime}), the nodes for state $t$ on levels $1$
and $2$ are given rank $0$, like other $t$-nodes. In the absence of a path forcing this rank, their
retrospective rank is $2$. 
\end{exa}

\begin{figure}
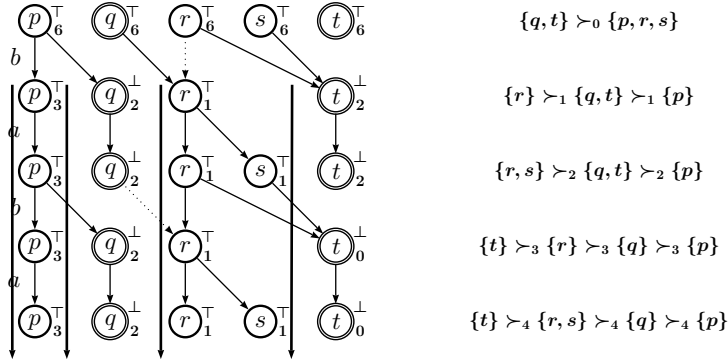

\centering
{
\begin{postscript}
\SmallPicture\VCDraw{\begin{VCPicture}{(-1,-12)(8,2)}
\def\level{0}\def\offset{0}
\State[p]{(0,\offset)}{p\level} \FinalState[q]{(2,\offset)}{q\level}
\State[r]{(4,\offset)}{r\level}
\State[s]{(6,\offset)}{s\level} \FinalState[t]{(8,\offset)}{t\level}
\def\level{1} \def\prevlevel{0} \def\offset{-2}
\State[p]{(0,\offset)}{p\level}\State[r]{(4,\offset)}{r\level}
\FinalState[q]{(2,\offset)}{q\level} \FinalState[t]{(8,\offset)}{t\level}
\EdgeR{p\prevlevel}{p\level}{b~~~~~~~}
\ChgEdgeLineStyle{dotted}\EdgeL{r\prevlevel}{r\level}{}\RstEdgeLineStyle
 \EdgeL{p\prevlevel}{q\level}{}
\EdgeL{q\prevlevel}{r\level}{} \EdgeL{r\prevlevel}{t\level}{} \EdgeL{s\prevlevel}{t\level}{}
\def\prevlevel{1} \def\level{2}\def\offset{-4}
\State[r]{(4,\offset)}{r\level} \State[p]{(0,\offset)}{p\level}
\FinalState[q]{(2,\offset)}{q\level} \State[s]{(6,\offset)}{s\level} \FinalState[t]{(8,\offset)}{t\level}
\EdgeR{p\prevlevel}{p\level}{a~~~~~~~} \EdgeL{r\prevlevel}{r\level}{} \EdgeL{q\prevlevel}{q\level}{}
\EdgeL{t\prevlevel}{t\level}{} \EdgeL{r\prevlevel}{s\level}{} 
\def\prevlevel{2} \def\level{3}\def\offset{-6}
\State[p]{(0,\offset)}{p\level} \State[r]{(4,\offset)}{r\level}
\FinalState[t]{(8,\offset)}{t\level} \FinalState[q]{(2,\offset)}{q\level} 
\EdgeR{p\prevlevel}{p\level}{b~~~~~~~} \EdgeL{r\prevlevel}{r\level}{} \EdgeL{p\prevlevel}{q\level}{}
\ChgEdgeLineStyle{dotted}\EdgeL{q\prevlevel}{r\level}{}\RstEdgeLineStyle
\EdgeL{r\prevlevel}{t\level}{} \EdgeL{s\prevlevel}{t\level}{}
\def\prevlevel{3} \def\level{4}\def\offset{-8}
\State[r]{(4,\offset)}{r\level} \State[p]{(0,\offset)}{p\level}
\FinalState[q]{(2,\offset)}{q\level} \State[s]{(6,\offset)}{s\level} \FinalState[t]{(8,\offset)}{t\level}
\EdgeR{p\prevlevel}{p\level}{a~~~~~~~} \EdgeL{r\prevlevel}{r\level}{} \EdgeL{q\prevlevel}{q\level}{}
\EdgeL{t\prevlevel}{t\level}{} \EdgeL{r\prevlevel}{s\level}{} 
\hide{
\def\prevlevel{4} \def\level{5}\def\offset{-10}
\State[r]{(4,\offset)}{r\level} \State[p]{(0,\offset)}{p\level}
\FinalState[q]{(2,\offset)}{q\level} \State[s]{(6,\offset)}{s\level} \FinalState[t]{(8,\offset)}{t\level}
\EdgeR{p\prevlevel}{p\level}{a~~~~~~~} \EdgeL{r\prevlevel}{r\level}{} \EdgeL{q\prevlevel}{q\level}{}
\EdgeL{s\prevlevel}{s\level}{} \EdgeL{t\prevlevel}{t\level}{} \EdgeL{r\prevlevel}{s\level}{} 
\def\prevlevel{5} \def\level{6}\def\offset{-12}
\State[p]{(0,\offset)}{p\level} \State[r]{(4,\offset)}{r\level}
\FinalState[t]{(8,\offset)}{t\level} \FinalState[q]{(2,\offset)}{q\level} 
\EdgeR{p\prevlevel}{p\level}{b~~~~~~~} \EdgeL{r\prevlevel}{r\level}{} \EdgeL{p\prevlevel}{q\level}{}
 \EdgeL{r\prevlevel}{t\level}{} \EdgeL{s\prevlevel}{t\level}{}
\ChgEdgeLineStyle{dotted}\EdgeL{q\prevlevel}{r\level}{}\RstEdgeLineStyle
}

\psline[linewidth=2pt] (-0.60,-1.7)(-0.60,-9)
\psline[linewidth=2pt] (0.85,-1.7)(0.85,-9)
\psline[linewidth=2pt] (3.35,-1.7)(3.35,-9)
\psline[linewidth=2pt] (6.84,-1.7)(6.82,-9)
\begin{boldmath}
\multirput(0.60,0.25)(2.00,0.0){5}{\large$\top$}
\multirput(0.60,-1.75)(0.0,-2.0){4}{\large$\top$}
\multirput(2.65,-1.60)(0.0,-2.0){4}{\large$\bot$}
\multirput(4.60,-1.75)(0.0,-2.0){4}{\large$\top$}
\multirput(6.60,-3.75)(0.0,-4.0){2}{\large$\top$}
\multirput(8.65,-1.60)(0.0,-2.0){4}{\large$\bot$}

\multirput(0.63,-0.25)(2.03,0.0){2}{\large$6$} \multirput(6.63,-0.25)(2.03,0.0){2}{\large$6$}
\multirput(4.66,-0.13)(2.03,0.0){1}{\large$6$}
\multirput(0.60,-2.25)(0.0,-2.0){4}{\large$3$}
\multirput(2.65,-2.25)(0.0,-2.0){4}{\large$2$}
\multirput(4.60,-2.25)(0.0,-4.0){2}{\large$1$} \multirput(4.60,-8.25)(0.0,-4.0){1}{\large$1$}
\multirput(4.65,-4.12)(0.0,-6.0){1}{\large$1$}
\multirput(6.60,-4.25)(0.0,-2.0){1}{\large$1$} \multirput(6.60,-8.25)(0.0,-2.0){1}{\large$1$}
\multirput(8.65,-2.25)(0.0,-2.0){2}{\large$2$} \multirput(8.65,-6.25)(0.0,-2.0){2}{\large$0$}

\rput[u](15,0){\Large $\set{q,t}\succ_0\set{p,r,s}$}
\rput[u](15,-2){\Large $\set{r}\succ_1\set{q,t}\succ_1\set{p}$}
\rput[u](15,-4){\Large $\set{r,s}\succ_2\set{q,t}\succ_2\set{p}$}
\rput[u](15,-6){\Large $\set{t}\succ_3\set{r}\succ_3\set{q}\succ_3\set{p}$}
\rput[u](15,-8){\Large $\set{t}\succ_4\set{r,s}\succ_4\set{q}\succ_4\set{p}$}

\end{boldmath}

\end{VCPicture}}
\end{postscript}\qquad\qquad\qquad\qquad\qquad\qquad\qquad
}
\vspace{-0.5in}
\caption{The run \DAG \Gprime, where $0$ is an $F$-finite level.  The labels of $\lambda^0$ and
ranks in ${\bf r}^0$ are displayed as superscripts and subscripts, respectively.  The bold lines
display the sequences of $\top$-labeled classes in $\Gprime$.  The lexicographic order of states is
repeated on the right.  }\label{Fig:GPrime}
\end{figure}

We are now ready to define a new construction, generating an NBW $\A_L$, which combines the benefits
of the previous two constructions.  The automaton $\A_L$ guesses the $F$-finite level $k$, and uses
level rankings to check if the $k$-retrospective ranking is an odd ranking.  We partition the
operation of $\A_L$ into two stages. Until the level $k$, the NBW $\A_L$ is in the first stage,
where it deterministically tracks preordered subsets.  After level $k$, the NBW $\A_L$ moves to the
second stage, where it tracks ranks. This stage is also deterministic.  Consequently, the only
nondeterminism in $\A_L$ is indeed the guess of $k$. Before defining $\A_L$, we need some
definitions and notations. 

\hide{
Lemma~\ref{Slices_Provide_Ranking} gives us an alternative odd ranking to the one of \cite{KV01c}.
We now provide an NBW that uses level rankings to guess this retrospective ranking.  In comparison
to the slice-based construction in Definition \ref{Slice_Def}, employing level rankings gives us a
tighter bound and affords further improvements, discussed later. In comparison to the rank-based
construction in Definition \ref{KVDef}, this automaton needs make only a single guess on each
accepting run: what is the $F$-finite level $k$.  Before and after this guess we proceed
deterministically. Before $k$, our automaton tracks preordered subsets.  At level $k$ our automaton
moves to ranks, and these ranks proceed stably. We thus obtain an automaton with a linear-sized
transition relation. We now introduce the machinery to define the automaton.
}

Recall that $\PSQ$ denotes the set of preordered subsets of $Q$, and $\TLRM$ the set of tight level
rankings bounded by $m$.  We distinguish between three types of transitions of $\A_L$: transitions
within the first stage, transitions from the first stage to the second, and transitions within the
second stage.  The first type of transition is similar to the one taken in $\A_S$, by means of the
$\sigma$-successor relation between preordered subsets. Below we explain in detail the other two
types of transitions.
Recall that in the retrospective ranking ${\bf r}^k$, each class in $\Gprime$
labeled $\top$ by $\lambda^k$ is given a unique odd rank. Thus the rank of a node $u$ depends on the
number of $\top$-labeled classes larger than it, denoted $\prerank(u)$.

We begin with transitions where $\A_L$ moves between the stages: from a preordered subset
$\zug{S,\preceq}$ to a level ranking.  On level $k+1$, a node is labeled $\top$ iff it is a
non-$F$-node.  Thus for every $q \in S$, let
$\prerankstate(q) = \abs{\set{[v]\text{ \small$\mid$~} v \in S \setminus F, u \prec
v}}$ be the number of non-$F$-classes
larger than $q$. We now define $\torank \colon {\PSQ} \to {\TLRM}$. Let
$\torank(\zug{S,\preceq})$ be the tight level ranking $f$ where for every $q$: 
$$f(q)=\begin{cases}
\bot& \text{if }q \not \in S,\\
2\prerankstate(q)&\text{if }q \in S \cap F,\\
2\prerankstate(q)\splus1&\text{if }q \in S \setminus F.
\end{cases}$$

We now turn to transitions within the second stage, between level rankings.  The rank of a node $v$
is inherited from its predecessor $u$ in $\Gprime$. However, $\lambda^k$ may label a finite class $\top$. If a
$\top$-labeled class larger than $u$ has no children, then $\prerank(u) \geq \prerank(v)$. In this
case the rank of $v$ decreases.  Given a level ranking $f$, for every $q \in Q$ where $f(q) \neq
\bot$, let $\gamma(q) = \abs{\set{f(q')~|~q' \in Q,~f(q')\text{ is odd, } f(q')<f(q)}}$ be the
number of odd ranks in the range of $f$ lower than $f(q)$. We define the function $\tighten \colon \LR^m
\to \TLRM$. Let $\tighten(f)$ be  the tight level ranking $f'$ where for every $q$: 
\hide{if $f(q) = \bot$, then $f'(q) = \bot$; if $q \in F$, then $f'(q)=2\gamma(q)$; and if $q \not \in F$,
then $f'(q)=2\gamma(q)\splus 1$.}
{
$$f'(q') = \begin{cases}
\bot &\text{if }f(q) = \bot,\\
2\gamma(q)&\text{if }f(q) \neq \bot \text { and } q \in F,\\
2\gamma(q)\splus 1&\text{if }f(q) \neq \bot \text { and } q \not \in F.
\end{cases}$$
}
Note that if $f$ is tight, then $f'=f$, and that while $\tighten$ may merge
two even ranks, it cannot merge two odd ranks.

For a level ranking $f$, a letter $\sigma \in \Sigma$, and $q' \in Q$, let
$\pred(q',\sigma,f)=\{q\mid f(q) \neq \bot,~q' \in \rho(q,\sigma)\}$ be the predecessors of $q'$
given a non-$\bot$ rank by $f$. The lowest ranked element in this set corresponds to the
predecessor in $\G$ with the maximal profile. With two exceptions, $q'$ will inherit this lowest rank. 
First, \tighten might shift the rank down. Second, if $q'$ is in $F$, it cannot be given
an odd rank. For $n \in \N$, let $\lfloor n\rfloor_{even}$ be: $n$ when $n$ is even; and $n\!-\!1$ when
$n$ is odd.  Define the \emph{$\sigma$-successor of $f$} to be $\tighten(f')$ where for every $q'
\in Q$:
$$f'(q') = \begin{cases} 
\bot & \text{if $\pred(q',\sigma,f) = \emptyset$},\\
\lfloor \min(\set{f(q)~|~q \in \pred(q',\sigma,f)})\rfloor_{even} & \text{if $\pred(q',\sigma,f) \neq \emptyset$ and $q' \in F$},\\
\min(\set{f(q)\mid q \in\pred(q',\sigma,f)}) & \text{if $\pred(q',\sigma,f) \neq \emptyset$ and $q' \not \in F$}.
\end{cases}$$

\begin{defi}\label{Slice_Rank}
For an NBW $\A = \zug{\Sigma, Q, Q^{in}, \rho, F}$, let $\A_L$ be the NBW\\
$\zug{\Sigma, {\PSQ} \cup  ({\TLRM} \times 2^Q), Q^{in}_{L}, \rho_{L}, {\TLRM} \times \set{\emptyset}}$, where
\begin{iteMize}{$\bullet$}
\item $Q^{in}_L = \set{\zug{Q^{in},\preceq^{in}}}$ where $\preceq^{in}$ is such that for all $q,r
\in Q^{in}$, $q \preceq r\text{ iff } q \not \in F \text{ or }r\in F$.
\item $\rho_L(\mcS,\sigma) = \set{\mcS'} \cup \set{\rzug{\torank(\mcS'),\emptyset} }$,
where $\mcS'$ is the $\sigma$-successor of $\mcS$.
\item \begin{tabbing}
$\rho_L(\rzug{f,O},\sigma) = \set{\rzug{f',O'}}$ 
           w\=here $f'$ is the $\sigma$-successor of $f$\\
           \>and $O'= \begin{cases}
\rho(O,\sigma) \setminus odd(f') & \text{if }O \neq \emptyset,\\
even(f')                         & \text{if }O=\emptyset.\\
\end{cases}$
\end{tabbing}\vspace{6 pt}
\end{iteMize}
\end{defi}

\noindent The proof of Theorem~\ref{SR_Complement} is based on Lemmas~\ref{Odd_Ranking_Rejecting} and
\ref{Slices_Provide_Ranking} and Corollary~\ref{Slices_Complement}.

\begin{theorem}\label{SR_Complement}
For every NBW $\A$, it holds that $L(\A_L)=\overline{L(\A)}$.
\end{theorem}
\begin{proof}
\cbstart
Consider a word $w \in \Sigma^\omega$ and the run \DAG $\G$. We first make the following claims
about every infinite run
$\zug{S_0,\preceq_0},\ldots,\zug{S_{k},\preceq_{k}},\rzug{f_{k+1},O_{k+1}},\rzug{f_{k+2},O_{k+2}},\ldots$.
For $i > k$, define $S_i=\set{q\mid f_{i}(q)\neq \bot}$.

\begin{enumerate}[(1)]

\item\label{Step:psS} \textit{The states in $S_i$ are precisely $\set{q\mid \rzug{q,i} \in \G}$.}\\
This follows by the definitions of $\sigma$-successors of preordered subsets and 
$\sigma$-successors of level rankings.

\item\label{Step:pPreorder} \textit{The preorder $\preceq_i$ is the projection of $\preceq$
onto states occurring at level $i$.} \\
This follows from Lemma~\ref{Lexicographic_Ordering} and the definition of $\sigma$-successors. 

\item\label{Step:pTransitions} \textit{For every $i \leq k$, state $q \in S_i$, and $s \in S_{i+1}$, it
holds that $s \in \rho_{\zug{S_i,\preceq_i}}(q,\sigma_i)$ iff $E'(\zug{q,i},\zug{s,i\splus1}).$}\\
This follows from the definitions of $E'$ and $\rho_{\zug{S_i,\preceq_i}}$. 

\item\label{Step:sTransitions} \textit{For every $i > k$ and $q,s \in S_i$, if $f_i(q) >
f_i(s)$, then $\zug{q,i} \prec_i \zug{s,i}$.}

\item\label{Step:sOddTight} \textit{For every $i > k$ and $q,s \in S_i$, if $f_i(s)$ is
odd and $\zug{q,i} \prec_i \zug{s,i}$, then $f_i(q) > f_i(s)$.}\\
This and \ref{Step:sTransitions} are proven below.

\item\label{Step:OddBottom} \textit{For every $i \geq k$ and $q \in S_i$, it holds that $f_i(q)$ is
even iff $\lambda^k(\zug{q,i})=\bot$. }\\
This follows from the definition of $\lambda^k$, which assigns $\bot$ to $F$-nodes and their
descendants in $\Gprime$, and $f_i$, which assigns even ranks to states in $F$.  By
\ref{Step:sTransitions}, the parent of a node in $\Gprime$ will be the parent with the lowest rank.
Thus the descendants of $F$-nodes in $\Gprime$ will inherit the even
rank of their parent.\vspace{6 pt}
\end{enumerate}

\noindent We simultaneously prove \ref{Step:sTransitions} and \ref{Step:sOddTight} by induction.  As a base
case, both hold from the definition of $\torank$. As the inductive step, assume both hold for level
$i$.
To prove step~\ref{Step:sTransitions}, take two states $q,s \in S_{i+1}$ where $f_{i+1}(q) >
f_{i+1}(s)$.  Each state has a parent in $\Gprime$, i.e. a $q'$ and $s'$ so that $E'(q',q)$ and
$E'(s',s)$. By the inductive hypothesis, this implies 
$f_i(q')=\min(\set{f_i(q')\mid q \in \rho(q',\sigma_i)})$ and $f_i(s') = \min(\set{f_i(s')\mid s \in
\rho(s',\sigma_i)})$.  We analyze two cases. When $f_i(q') > f_i(s')$, by the inductive hypothesis
we have $\zug{q',i} \prec_i \zug{s',i}$.  Since $E'(q',q)$ and $E'(s',s)$, by
Lemma~\ref{Gprime_Captures_Profiles} this implies $\zug{q,i\splus1} \prec_{i\splus 1} \zug{s,i\splus1}$.  Alternately,
when $f_i(q') = f_i(s')$, then for $f_{i+1}(q)>f_{i+1}(s)$ to hold, it must be that $f_i(q')$ is
odd, $s \in F$, and $q \not \in F$. Since $f_i(q')=f_i(s')$ is odd, by the inductive hypothesis we
have that $\zug{q',i} \equiv \zug{s',i}$.  By Lemma~\ref{Gprime_Captures_Profiles} we then have
$h_{\zug{q,i\splus1}}=h_{\zug{q',i}}0 < h_{\zug{s,i\splus1}}=h_{\zug{s',i}}1$.

To prove step~\ref{Step:sOddTight}, consider when $f_{i+1}(s)$ is odd and $\zug{q,i\splus1} \prec
\zug{s,i\splus1}$. This implies that $h_\zug{s,i\splus1}=h_\zug{s',i}0$. Thus in order for 
$\zug{q,i\splus1} \prec_{i\splus 1} \zug{s,i\splus1}$ to hold, $\zug{q',i} \prec_i \zug{s',i}$ must hold. By the inductive hypothesis, this
implies $f_i(q') > f_i(s')$. Before the $\tighten$ function reduces ranks, since 
$f_{i+1}(q) = \lfloor f_i(q') \rfloor_{even}$, and $f_{i+1}(s)$ is odd, it must be that $f_{i+1}(q) > f_{i+1}(s)$.
The \tighten function can shift $f_{i+1}(q)$ down more than $f_{i+1}(s)$ only when an odd rank
between $f_{i+1}(s)$ and $f_{i+1}(q)$ becomes empty. Since this odd rank must be two greater than
$f_{i+1}(s)$, reducing $f_{i+1}(q)$ by 2 cannot change that $f_{i+1}(q) > f_{i+1}(s)$.  We now proceed with the proof of Theorem~\ref{SR_Complement}. 

In one direction, assume the run
$\zug{S_0,\preceq_0},\ldots,\zug{S_k,\preceq_k},$ $\rzug{f_{k+1},O_{k+1}},\rzug{f_{k+2},O_{k+2}},\ldots$ is accepting.  We construct
a ranking ${\bf r}$ of $\G$ as follows. For all nodes $u$ on level $i \leq k$, ${\bf r}(u)=m$. For all nodes
$\zug{q,i}$ where $i > k$, ${\bf r}(\zug{q,i})=f_i(q)$.  We note that each state is given at most the
minimum rank of all its parents, and that no state in $F$ is given an odd rank, thus ${\bf r}$ is in fact
a ranking. That ${\bf r}$ is an odd ranking follows from the cut-point construction.

In the other direction, assume $\G$ is a rejecting run \DAG. By Lemma~\ref{Slices_Provide_Ranking}
there exists a $k$ so that ${\bf r}^k$ is an odd ranking. We construct a run
$\mcS_0,\ldots,\mcS_k,\rzug{f_{k+1},O_{k+1}},\rzug{f_{k+2},O_{k+2}},\ldots$, which is uniquely defined
by the transition relation of Definition~\ref{Slice_Rank}. Further, the transition relation of
Definition~\ref{Slice_Rank} is total, so this run is infinite. To demonstrate that this run is
accepting, we will prove below that for every $i > k$ and $q \in S_i$, it holds that $f_i(q)={\bf
r}^k(\zug{q,i})$.  Since ${\bf r}^k$ is an odd ranking and the cut-point construction is identical
to that of Definition~\ref{KVDef}, this is sufficient to show the run is accepting.

Recall that if $\lambda(\zug{q,i})=\bot$, then ${\bf r}^k(\zug{q,i})=2\prerank(\zug{q,i})$, and otherwise
${\bf r}^k(\zug{q,i})=2\prerank(\zug{s,i})\splus1$. We can thus use \ref{Step:OddBottom} to simplify our claim. 
It suffices to show that for every $i > k$ and $q \in S_i$, we have 
$\prerank(\zug{q,i})= \lfloor f_{i}(q)/2 \rfloor$.  We proceed by induction over $i > k$.  As the
base case, consider a node $\zug{q,k}$. Recall that $\prerank(\zug{q,k})=
\abs{\set{[v]\text{~\small\textbar~} \lambda^k(v)=\top, \zug{q,k} \prec_k v}}$. By
the definition of $\lambda^k$, a node on level $k$ is labeled $\bot$ only when it is an $F$-node.
All other nodes inherit the label of their parents, and every node on level $k$ is $\top$-labeled.
From \ref{Step:pPreorder}, we then have that 
$\prerank(\zug{q,k+1})= \abs{\set{[v]\text{~\small\textbar~} v \in S \setminus F, u \prec v}}$,
which is the definition of $\prerankstate(q)=\lfloor f_{i}(q)/2 \rfloor$.

Inductively, assume the claim holds for every $q \in S_i$. We show for every $s \in S_{i+1}$, it
holds that $\prerank(\zug{s,i\splus1})= \lfloor f_{i+1}(s)/2 \rfloor$.
Let $q$ be the parent of $s$ in
$\Gprime$, i.e. $E'(q,s)$.  Take the set $P=\set{[v]\mid \lambda^k(v)=\top, \zug{q,i} \prec_i v}$.  of
$\top$-labeled equivalence classes greater than $Q$, By the inductive hypothesis, 
$\lfloor f_{i}(q)/2 \rfloor= \prerank(\zug{q,i})= \abs{P}$.  By the definition of ${\bf r}^k$, each 
$[v] \in P$ has a unique odd rank assigned to each of its elements.  By \ref{Step:sOddTight}, for each $[v]$
this odd rank is smaller than $f_i(q)$. Consider the subset of $P$
given by
$P_s = \{[v]\mid [v] \in P,\break [v]\hbox{ has $\top$-labeled child
    class on level }i\splus1\}$.  Define
$P_e=P\setminus P_s$ to be the complementary set: pipes that die on level $i$. By
\ref{Step:sOddTight}, before the $\tighten$ operation is applied, every element of $P_e$ has a
corresponding odd rank that is unoccupied on level $i\splus1$.  Since $q$ is clearly not in an element of
$P_e$, this odd rank must be less than $\lfloor f_{i}(q) \lfloor_{even}$.  Thus the final rank
assigned to $s$, after $\tighten$, is either
$f_{i}(q)-2\abs{P_e}$ or $\lfloor f_{i}(q)-2\abs{P_e} \rfloor_{even}$.  In both cases 
$\lfloor f_{i+1}(s)/2 \rfloor = \lfloor f_{i}(q)/2 \rfloor - \abs{P_e}$.  By the inductive hypothesis this is
equivalent to $\prerank(\zug{q,i}) - \abs{P_e} = \abs{P}-\abs{P_e}$.  By the definition of $P_s$ and
$P_e$, $\abs{P} - \abs{P_e} = \abs{P_s}$.  By Lemma~\ref{Gprime_Captures_Profiles}, every
$\top$-labeled child of a class in $P_s$ is lexicographically larger than $\zug{s,i\splus1}$. As every
$\top$-labeled child must have a unique parent in $P_s$, we conclude that
$\abs{P_s}=\prerank(\zug{s,i\splus1})$.
\cbend
\end{proof}

\noindent {\bf Analysis:} Like the tight-ranking construction in Section \ref{Sect:Tight}, the
automaton $\A_L$ operates in two stages. In both, the second stage is the set of tight level
rankings and obligation sets. The tight-ranking construction uses
sets of states in the first stage, and is bounded by the size of the second stage: $(0.96n)^n$
\cite{FKV06}.  The automaton $\A_L$ replaces the first stage with preordered subsets.
As the number of preordered subsets is $O((\frac{n}{e\ln 2})^n) \approx (0.53n)^n$ \cite{Var80}, the
size of $\A_L$ remains bounded by $(0.96n)^n$. This can be improved to $(0.76n)^n$: see
Section~\ref{App:Variants} and \cite{Sch09}. 
Further, $\A_L$ has a very restricted transition relation: states in the first stage only
guess whether to remain in the first stage or move to the second, and have nondeterminism of
degree 2. States in the second stage are deterministic.  Thus the transition relation is linear in
the number of states and size of the alphabet, and $\A_L$ is deterministic in the limit.

\section{Variations on $\A_{L}$}\label{App:Variants}
\cbstart

In this section we present two variations of $\A_{L}$: one based on Schewe's variant of the
rank-based construction that achieves a tighter bound; and one that is amenable to Tabakov and
Vardi's symbolic implementation of the rank-based construction.
Schewe's construction alters the cut-point of the rank-based construction to check only one even
rank at a time. Doing so drastically reduces the size of the cut-point: intuitively, we can avoid
carrying the obligation set explicitly. Instead we could carry the current rank $i$ we are checking,
and add to the domain of our ranking function a single extra symbol $c$ that indicates the state is
currently being checked, and thus is of rank $i$.  For an analysis of the resulting state space,
please see \cite{Sch09}. For clarity , we do not remove the obligation set from the construction.
Instead, states in this variant of the automaton carry with them the index $i$, and in a state
$\zug{f,O,i}$, it holds that $O \subseteq \set{q\mid f(q)=i}$. For a level ranking $f$, let
$\maxrank(f)$ be the largest rank in $f$. Note that $\maxrank(f)$, for a tight ranking, is always
odd.

\begin{defi}\label{Stable_Rank_Schewe}
For an NBW $\A = \zug{\Sigma, Q, Q^{in}, \rho, F}$, let $\A_{Schewe}$ be the NBW\\
$\zug{\Sigma, {\PSQ} \cup  ({\LR}^m \times 2^Q \times N), Q^{in}_{L}, \rho_{Sch}, F_{Sch}}$, where
\begin{iteMize}{$\bullet$}
\item $\rho_{Sch}(\mcS,\sigma) = \set{\rzug{\torank(\mcS'), \emptyset, 0} } \cup \set{\mcS'}$,
where $\mcS'$ is the $\sigma$-successor of $\mcS$.
\item \begin{tabbing}
$\rho_{Sch}(\rzug{f,O,i},\sigma) = \{\rzug{f',O',i'}\}$ \=where\\
           \>$f'$ is the $\sigma$-successor of $f$\\
           \>$i'=\begin{cases}
                i & \text{ if }O \neq \emptyset,\\
                (i\splus2)\bmod(\maxrank(f')\splus1)&  \text{ if } O = \emptyset,\\
           \end{cases}$\\
           \>and $O'= \begin{cases}
\rho(O,\sigma) \setminus odd(f') & \text{if }O \neq \emptyset,\\
\set{q\mid f'(q)=i'}                         & \text{if }O=\emptyset.\\
\end{cases}$
\end{tabbing}
\item $F_{Sch} = {\LR}^m \times \set{\emptyset} \times \set{0}$
\end{iteMize}
\end{defi}

\begin{theorem}\label{Schewe_Complement}
For every NBW $\A$, it holds that $L(\A_{Schewe})=\overline{L(\A)}$.
\end{theorem}
\begin{proof}
Given a word $w$, we relate the runs of $\A_{Schewe}$ and $\A_L$. As both
automata are comprised of two internally deterministic stages, with a
nondeterministic transition, each index $k$ defines a unique run for each
automaton.  As the first stage of both automata are identical, and the second
stage is deterministic, given a fixed $k$ let 
\[p_L=\zug{S_0,\preceq_0},\ldots,\zug{S_{k},\preceq_{k}},\rzug{f_{k+1},O_{k+1}},\rzug{f_{k+2},O_{k+2}},\ldots\]
be the run of $\A_L$ on $w$ that moves to the second stage on the $k$th transition, and let
the corresponding run of $\A_{Schewe}$ be
\[p_{Sch}\zug{S_0,\preceq_0},\ldots,\zug{S_{k},\preceq_{k}},\rzug{f'_{k+1},O'_{k+1},n_{k+1}},\rzug{f'_{k+2},O'_{k+2},n_{k+2}},\ldots\]
We show that $p_L$ is accepting iff $p_{Sch}$ is accepting, or more precisely that
$p_L$ is rejecting iff $p_{Sch}$ is rejecting.  First, we note that
the level rankings $f_{k+1},f_{k+2},\ldots$ and $f'_{k+1},f'_{k+2},\ldots$ in both automata are
defined by $\torank(\zug{S_k,\preceq_k})$ and the $\sigma$-successor relation, and thus for every $j
> k$, it holds $f_j=f'_j$. 

In one direction, assume that $p_{Sch}$ is rejecting.  This implies there
is some $j > k$ so that for every $j' > j$, $O'_{j'}$ is non-empty. In turn,
this implies that there is a sequence $q_j,q_{j+1},\ldots$ of states so that,
for every $j' \geq j$, we have that $q_{j'} \in O'_{j'}$, that
$f_{j'}(q_{j'})=n_j$, and that $q_{j'+1} \in \rho(q_{j'},w_{j'})$. If there is
no $l > j$ where $O_{l}=\emptyset$, then we have that $p_L$ is rejecting.
Alternately, if there is such a $l > j$, then $q_{l+1} \in O_{l+1}$, and for
every $l' > l$ we have $q_{l'} \in O_{l'}$. Again, this implies $p_L$ is
rejecting.

In the other direction, assume that $p_L$ is rejecting. This implies there is some $j > k$ so that for
every $j' > j$ the set $O_{j'}$ is non-empty. In turn, this implies that there is an
even rank $i$ and sequence $q_j,q_{j+1},\ldots$ of states so that, for every
$j' \geq j$, we have that $q_{j'} \in O_{j'}$, that $f_{j'}(q_{j'})=i$, and
that $q_{j'+1} \in \rho(q_{j'},w_{j'})$. We now consider the indexes $n_{j'}$
in $p_{Sch}$. If there is some $j' > j$ where $n_{j'}=i$, then for every $l\geq
j'$, it will hold that $q_{l} \in O'_{l}$, and $p_{Sch}$ will be rejecting.
Alternately, if there is no $j' > j$ where $n_{j'}=i$, then it must be that the
indexes $n_{j'}$ stops cycling through the even ranks. This implies the
obligation set stops emptying, and therefore that $p_{Sch}$ must be rejecting.
\end{proof}

To symbolically encode a deterministic-in-the-limit automaton, we avoid storing the preorders.  To
encode the preorder in a BDD as a relation would require a quadratic number of variables, increasing
the size unacceptably. Alternately, we could associate each state with its index in the preorder.
Unfortunately, calculating the index of each state in the succeeding preorder would require a global
compacting step, to remove indices that had become empty. To handle this difficulty, we simply store
only the subset in the first stage, and transition to an arbitrary level ranking when we move to the
second stage. This maintains determinism in the limit, and cannot result in false accepting run: we
can always construct an odd ranking from the sequence of level rankings. The construction and a
small example encoding are provided below.

\begin{defi}\label{Symbolic_DetLim}
For an NBW $\A = \zug{\Sigma, Q, Q^{in}, \rho, F}$, let $\A_{Symb}$ be the NBW\\
$\zug{\Sigma, 2^Q \cup ({\LR}^m \times 2^Q), Q^{in}, \rho_{Symb}, {\LR}^m \times \set{\emptyset}}$, where
\begin{iteMize}{$\bullet$}
\item $\rho_{Symb}(S,\sigma) = \set{\rho(S,\sigma)} \cup \set{\rzug{f, \emptyset} \mid f \in {\LR}^m \text{ and for all } q \in Q,~
f(q) \neq \bot \text{ iff } q \in \rho(S,\sigma) }$.
\item $\rho_{Symb}(\rzug{f,O},\sigma) = \rho_L(\rzug{f,O},\sigma)$
\end{iteMize}
\end{defi}

\begin{theorem}\label{Symb_Complement}
For every NBW $\A$, it holds that $L(\A_{Symb})=\overline{L(\A)}$.
\end{theorem}
\begin{proof}
In one direction, assume $w \in \overline{L(\A)}$. This implies $w \in L(\A_L)$, and thus there exists an
accepting run
$\zug{S_0,\preceq_0},\ldots,\zug{S_{k},\preceq_{k}},\rzug{f_{k+1},O_{k+1}},\rzug{f_{k+2},O_{k+2}},\ldots$
of $\A_L$ on $w$. We show that $S_0,\ldots,S_k,\rzug{f_{k+1},O_{k+1}},\rzug{f_{k+2},O_{k+2}},\ldots$
is an accepting run of $\A_{Symb}$ on $w$. We note that in the second stage transitions and
accepting states in $\A_{Symb}$ are identical to $\A_L$. Thus to show that this is an accepting run
$\A_{Symb}$, we only need show that the run is valid from $0$ to $k+1$,

By definition, $S_0=Q^{in}$ is the initial state of $\A_{Symb}$. For every $i,~0 \leq i <
k$, it holds that $S_{i+1}=\rho(S_i,w_i) \in \rho_{Symb}(S_i,w_i)$. Finally, consider the transition
from $S_k$ to $\rzug{f_{k+1},O_{k+1}}$. Let $\zug{S_{k+1},\preceq_{k+1}}$ be
the $\sigma$-successor of $\zug{S_k,\preceq_k}$.  By definition, $S_{k+1}=\rho(S_k,w_k)$. 
By the transition relation of $\A_L$,  we have
$f_{k+1}=\torank(\zug{S_{k+1},\preceq_{k+1}})$ and $O_{k+1}=\emptyset$.  By the definition of
$\torank$, for every $q \in Q$ it holds that $f_{k+1}(q)=\bot$ iff $q \not \in S_{k+1}$.
Thus 
$\zug{f_{k+1},O_{k+1}} \in \rho_{Symb}(S_k)$, and
$S_0,\ldots,S_k,\rzug{f_{k+1},O_{k+1}},\rzug{f_{k+2},O_{k+2}},\ldots$ is an accepting run of
$\A_{Symb}$ on $w$. 

In the other direction,  if $w \in L(\A_{Symb})$, there is an accepting run
$S_0,\ldots,S_k,\rzug{f_{k+1},O_{k+1}},$ $\rzug{f_{k+2},O_{k+2}},\ldots$ of $\A_{Symb}$ on $w$.
From this run we construct an odd ranking of $\G$, which implies $w \in \overline{L(\A)}$.  Define
the ranking function $\bf{r}$ so that for every $\rzug{q,i} \in \G$: if $i \leq k$ then
$\textbf{r}(\rzug{q,i})=m=2\abs{Q \setminus F}$; and if $i > k$ then
$\textbf{r}(\rzug{q,i})=f_i(q)$.  As demonstrated in the proof of Theorem \ref{SR_Complement}, the
definition of $\sigma$-successors and $\G$ implies that when $i > k$, it holds that $f_i(q) \neq
\bot$. Similarly, by the definition of $\sigma$-successors no path in $\G$ can increase in rank under
$\textbf{r}$. We conclude that $\textbf{r}$ is a valid ranking function.

To demonstrate that $\textbf{r}$ is an odd ranking, assume by way of
contradiction that there is a path $\zug{q_0,0}, \zug{q_1,1}, \ldots$ in $\G$
that gets trapped in an even rank.  Let $j$ be the point at which this path
gets trapped, or $k+1$, whichever is later. If there is no $j' > j$ such that
$O_{j'} = \emptyset$, then there is no accepting state after $j$, and the run
would not be accepting. If there is such a $j'$, then $O_{j'+1}$ would contain
$q_{j'+1}$, as $f_{j'+1}(q_{j'+1})$ is even. At every point $j'' > j'+1$, the
obligation set will contain $q_{j''}$, and thus there will be no accepting
state after $j'$, and the run would not be accepting.  However, we have that
the run is accepting as a premise. Therefore no path in $\G$ gets trapped in
an even rank, $\textbf{r}$ is an odd ranking, and by Lemma \ref{Odd_Ranking_Rejecting} $w \in \overline{L(\A)}$.
\end{proof}

As an example, Figure \ref{Fig:SMV_Encoding} is the SMV encoding of the complement of a two-state automaton.

\fvset{formatcom=\scriptsize}
{
\begin{SaveVerbatim}{SMVExample}
typedef STATE 0..1;  /* Size for complemented automaton: 2, maximum allowed rank = 2*/
module main() { 
 letter: {a,b};                   /* The transition letter */
 rank: array STATE of 0..3;       /* The value 3 represents bottom */
 phase : 0..1;                    /* The phase of the automaton, ranks 2 or 3 in phase 0*/
 subset: array STATE of boolean;  /* The obligation set vector */
 init(rank) := [2,2,2,2];         /* 2 to initial states, 3 to others */
 init(subset) := [1,1,1,1];       /* initially rejecting */
 init(phase) := 0;
 next(phase) := {i : i=0..1, i >= phase}; 

 /* Define the rank of states in the next time step. Cases fall through. */
 /* state 0 has transition from 0 on a and b */
 next(rank[0]) := case {
     rank[0]=3 : 3;
     next(phase)=0 : 2; 
     phase=0 & next(phase)=1 : {i : i=0..2, i <= rank[0]};
     phase=1 : rank[0];
 };

 /* 1 has transition from 1 on a and from 0 on b. 1 is accepting */
 next(rank[1]) := case {
    letter=a & rank[1]=3 : 3;
    letter=a & next(phase)=0 : 2;
    letter=a & phase=0 & next(phase)=1 : {i : i=0..2, i <= rank[1] & i in {0,2}};
    letter=a & phase=1 : {i : i=0..2, i in {rank[1], rank[1]-1} & i in {0,2}};
    letter=b & rank[0]=3 : 3;
    letter=b & next(phase)=0 : 2;
    letter=b & phase=0 & next(phase)=1 : {i : i=0..2, i <= rank[0] & i in {0,2}};
    letter=b & phase=1 : {i : i=0..2, i in {rank[0], rank[0]-1} & i in {0,2}};
 };

 
 /* Defining the transitions of the P-set */
 if (next(phase)=0)  {
     forall (i in STATE) next(subset[i]) := 1;
 } else {
   if (subset=[0,0,0,0]) { /* The P-set is empty */
     forall (i in STATE) next(subset[i]) := next(rank[i]) in {0,2};
   } else { /* The P-set is non-empty */
     if (letter=a) {
       next(subset[0]) := (subset[0]) & next(rank[0]) in {0,2};
       next(subset[1]) := (subset[1]) & next(rank[1]) in {0,2};
     } else { /* letter=b */
       next(subset[0]) := (subset[0]) & next(rank[0]) in {0,2};
       next(subset[1]) := (subset[0]) & next(rank[1]) in {0,2};
 }}} 
SPEC 0;
 FAIRNESS subset=[0,0,0,0];
}
\end{SaveVerbatim}
}

\begin{figure}[htbp]
	\fbox{
		\begin{minipage}{\textwidth}
			\BUseVerbatim{SMVExample}
		\end{minipage}
	}
			\caption{The SMV encoding of the $\A_{Symb}$, for the two-state
automaton consisting of states $p$ and $q$ of Figure \ref{Fig:Automaton}.}
			\label{Fig:SMV_Encoding}
\end{figure}

\cbend

\section{Discussion}

We have unified the slice-based and rank-based approaches by phrasing the former in the language of
run \DAGs. This enables us to define and exploit a retrospective ranking, providing a
deterministic-in-the-limit complementation construction that does not employ determinization.
Experiments show that the more deterministic automata are, the better they perform in practice
\cite{ST03}. By avoiding determinization, we reduce the cost of such a construction from $(n^2/e)^n$
to $(0.76n)^n$ \cite{Pit06}.

In addition, our transition generates a transition relation that is linear in the number of states
and size of the alphabet. Schewe demonstrated how to achieve a similar linear bound on the
transition relation, but the resulting relation is larger and is not deterministic in the limit
\cite{Sch09}.

As shown in Section~\ref{App:Variants}, the use of level rankings affords several improvements from
existing research on the rank-based approach. First, the cut-point construction of Miyano and
Hayashi \cite{MH84} can be improved. Schewe's construction only checks one even rank at a time,
reducing the size of the state space to $(0.76n)^n$, only an $n^2$ factor from the lower bound
\cite{Sch09}. As Schewe's approach does not alter the progression of the level rankings, it can be
applied directly to the second stage of Definition~\ref{Slice_Rank}.  The resulting construction
inherits the asymptotic state-space complexity of \cite{Sch09}. Second, symbolically encoding a
preorder is complicated.  In contrast, ranks are easily encoded, and the transition between ranks is
nearly trivial to implement in SMV \cite{TV07}.  By changing the states in first stage of $\A_L$
from preordered subsets to simple subsets, and guessing the appropriate transition to the second
stage, we obtain a symbolic representation while maintaining determinism in the limit.  This
approach sacrifices the linear-sized transition relation, but this is less important in a symbolic
encoding.  Finally, although not addressed in Section~\ref{App:Variants}, the subsumption relations
of Doyen and Raskin \cite{DR09} could be applied to the second stage of 
Definition~\ref{Slice_Rank}, while it is unclear if it could be applied at all to the slice-based
construction. 

From a broader perspective, we find it very interesting that the prospective and retrospective
approaches are so strongly related. Odd rankings seem to be inherently ``prospective,'' depending on
the descendants of nodes in the run \DAG. By investigating the slice-based approach, we are able to
pinpoint the dependency on the future to a single component: the $F$-free level. This suggests it
may be possible to use odd rankings for determinization, automata with other accepting conditions,
and automata on infinite trees.

\bibliographystyle{alpha}
\bibliography{ok,cav,sfogarty}

\newpage
\appendix

\section{Slices}\label{App:Slices}

The paper of \kahler et al. introduces the notion of the split tree, reduced
split tree, and skeleton of an automaton $\A$ and word $w$
 Trees are represented as
prefix-closed non-empty subsets of $\set{0,1}^*$.  In a tree $V$, a node $v0$
is called the left child of $v$, and $v1$ the right child of $v$. The root is
$\epsilon$. A node $v$ is said to be on level $i$ when $\abs{v}=i$. For a set
$L$, an $L$-labeled tree is a pair $\zug{V,l}$ where $V$ is a tree and $l : V
\to L$ is a label function. By abuse of notation, for an $L$-labeled tree
$T=\zug{V,l}$ and vertex $v$, say $v \in T$ when $v \in V$, and let
$T(v)=l(v)$. For two nodes $v$ and $v'$, say that $v' > v$ when
$\abs{v}=\abs{v'}$ and $v'$ is to the right of, i.e. lexicographically larger
than, $v$.

The \emph{split tree}, written $T^{sp}$, is the $2^{Q}$-labeled tree defined
inductively as follows\footnote{Compared to \cite{KW08}, these definitions 
reverse the left and right children.  This was done to match 
the paper.}.  As a base case, $\epsilon \in T^{sp}$ and
$T^{sp}(\epsilon)=Q^{in}$.  Inductively, given a node $v$ on level $i$, let
$P=T^{sp}(v)$.  If $\rho(P,w_i) \setminus F \neq \emptyset$ then $v0 \in
T^{sp}$ and $T^{sp}(v0)= \rho(P,w_i) \setminus F$.  Similarly, if $\rho(P,w_i)
\cap F \neq \emptyset$, then $v1 \in T^{sp}$ and $T^{sp}(v1) = \rho(P,w_i) \cap
F$.  As argued in \cite{KW08}, branches in $T^{sp}$ correspond to runs of $\A$
on $w$.  We gloss over this discussion and simply state that $w \in L(\A)$ iff
$T^{sp}(\A,w)$ has a branch that goes right infinitely often. 
The split tree is analogous to $G_w$. Each path $p$ to node $\rzug{q,i} \in
G_w$ corresponds to a node $v$ on level $i$ of $T^{sp}$ that contains $q$ in
its label. Edges in $G_w$ correspond to edges in $T^{sp}$, and thus paths in
$G_w$ correspond to paths in $T^{sp}$.

\begin{lem}\label{Split_Nodes_Correspond}
For every state $q$ and level $i$, $\rzug{q,i} \in \G_w$ iff there is at least one node $v \in
T^{sp}$ where $\abs{v}=i$ and $q \in T^{sp}(v)$.
\end{lem}
\begin{proof}
We prove this by simple induction over $i$. As the base case we have that $\epsilon \in T^{sp}(v)$
and $T^{sp}(\epsilon)=Q^{in}$, while by definition $\zug{q,0} \in \G_w$ iff $q \in Q^{in}$. Thus our
lemma holds for $i=0$.

Inductively, assume that the lemma holds for $i=1$, and let $q' \in Q$. In one direction, if
$\rzug{q',i+1} \in \G_w$, then there is a run $p$ so that $p_{i+1}=q'$. By the inductive hypothesis,
there is a node $v \in T^{sp}$ where $\abs{v}=i$ and $p_i \in T^{sp}(v)$.  If $q' \not \in F$, then
$\rho(T^{sp}(v),w_i) \setminus F \neq \emptyset$, $v0 \in T^{sp}$, and $q' \in T^{sp}(v0) =
\rho(T^{sp}(v),w_i) \setminus F$.  Similarly, if $q' \in F$, then $\rho(T^{sp}(v),w_i) \cap F \neq
\emptyset$, $v1 \in T^{sp}$, and $q' \in T^{sp}(v1) = \rho(T^{sp}(v),w_i) \cap F$.

In the other direction, if there is a node $v' \in T^{sp}$ so that $\abs{v'}=i+1$ and $q' \in
T^{sp}(v)$, then $v'$ has a parent $v$ so that $\abs{v}=i$. As  $q' \in \rho(T^{sp}(v),w_i)$, there
is a state $q \in T^{sp}(v)$ so that $q' \in \rho(q,w_i)$. By the inductive hypothesis, $\rzug{q,i}
\in \G_w$, and by definition $q \in \rho(Q^{in},w_0\cdots w_{i-1})$. By the definition of a run,
this implies $q' \in \rho(Q^{in},w_0\cdots w_i)$, and thus $\rzug{q',i+1} \in \G_w$.
\end{proof}

\begin{lem}\label{Split_Paths_Correspond}
For every $q, q'$, and $i$, it holds that $\zug{\rzug{q,i},\rzug{q',i+i}} \in
E$ iff there are nodes $v$ and $v'$ in $T^{sp}$ so that $\abs{v}=i$, $v'$ is
a child of $v$, $q \in T^{sp}(v)$, and $q' \in T^{sp}(v')$. 
\end{lem}
\begin{proof}
In one direction, let $q, q'$, and $i$ be such that $\zug{\rzug{q,i},\rzug{q',i+i}} \in E$. By the
definition of $E$, we have $q' \in \rho(q,w_i)$.  By Lemma \ref{Split_Nodes_Correspond}, there is a
node $v \in T^{sp}$ so that $\abs{v}=i$ and $q \in T^{sp}(v)$.  If $q' \not \in F$, then let $v' =
v0$, otherwise  $q' \in F$ and let $v' = v1$. In either case that $q' \in \rho(q,w_i)$ implies that
$v' \in T^{sp}$ and $q' \in T^{sp}(v')$. 

In the other direction, let $q, q'$, and $i$ be such that there are nodes $v$ and $v'$ in $T^{sp}$
where $\abs{v}=i$, $v'$ is a child of $v$, $q \in T^{sp}(v)$, and $q' \in T^{sp}(v')$. By Lemma
\ref{Split_Nodes_Correspond}, we have that $\rzug{q,i} \in \G_w,$ and $\rzug{q',i+1} \in \G_w$. By
the definition of $T^{sp}$, we have $q' \in \rho(q,w_i)$, and thus $\zug{\rzug{q,i},\rzug{q',i+i}} \in E$.
\end{proof}

The \emph{reduced split tree}, written $T^{rs}$, keeps only the rightmost 
instance of each state at each level of the tree. This bounds the width of
$T^{rs}$ to $n$. Formally, we define $T^{rs}$ inductively as follows. As a base
case, the root $\epsilon \in T^{rs}$, and $T^{rs}(\epsilon)=Q^{in}$.
Inductively, given a node $v \in T^{rs}$ on level $i$, let $P=T^{rs}(v)$ and
let $P'=\bigcup \set{\rho(T^{rs}(v') \mid v' \in T^{rs}\text{ and } v' < v}$.
If $(\rho(P,w_i) \setminus F) \setminus P' \neq \emptyset$ then $v0 \in T^{rs}$
and $T^{rs}(v0)= (\rho(P,w_i) \setminus F) \setminus P'$.  Similarly, if
$(\rho(P,w_i) \cap F) \setminus P' \neq \emptyset$, then $v1 \in T^{rs}$ and
$T^{rs}(v1) = (\rho(P,w_i) \cap F) \setminus P'$.
The reduced split tree is analogous to the profiles of nodes in $G_w$ and the edges in \Gprime.
Since paths in $G_w$ correspond to paths in $T^{sp}$, the lexicographically
maximal path through $G_w$ to a node $\rzug{q,i}$ corresponds to the rightmost
path through $T^{sp}$ to an instance of $q$ on level $i$.  This is the only
instance that remains in $T^{rs}$.

\begin{lem}\label{Reduced_Split_Nodes}
For every node $\rzug{q,i} \in \Gprime$, there is a node $v \in T^{rs}$ where
$\abs{v}=i$ and $q \in T^{rs}(v)$. Further, $h_{\rzug{q,i}}=0v$.
\end{lem}
\begin{proof}
By Lemma \ref{Split_Nodes_Correspond}, for every $\rzug{q,i}$. there is at least one node $v' \in
T^{rs}$ where $\abs{v'}=i$ and $q \in T^{rs}(v')$. Let $v$ be the rightmost such node. We must show
that $h_{\rzug{q,i}}=0v$, and we do so by induction over $i$. As a base case, we have that $q \in
Q^{in}$, $i=0$, and $v=\epsilon$. Since, by assumption, $Q^{in} \cap F = \emptyset$, we have
$h_\rzug{q,0}=0=0v$. Inductively, assume that this lemma holds for a fixed $i$, and let $q'$ be such
that $\rzug{q',i+1} \in \Gprime$. Let $b$ be $0$ if $q' \not \in F$, and $1$ if $q' \in F$. Since
there are no orphan nodes in $\Gprime$, we know that there is a $q$ on level $i$ such that
$\zug{\rzug{q,i},\rzug{q',i+1}} \in \Gprime$.  By Lemma \ref{Gprime_Captures_Profiles}, we know that
$h_\rzug{q',i+1}=h_\rzug{q,i}b$, and that $\rzug{q,i}$ has the lexicographically
maximal profile of all predecessors of $\rzug{q',i+1}$.   By the inductive hypothesis, there is a node $v \in T^{rs}$ so that
$\abs{v}=i$, $q \in T^{rs}(v)$, and $h_{\rzug{q,i}}=v$. Since lexicographic maximality in profiles
corresponds to being rightmost in the tree, this means $v$ is the rightmost node containing $q$ in
$T^{sp}$. Thus $vb$ is the rightmost node containing $q'$ in $T^{sp}$, and the only node containing
$q'$ in $T^{rs}$.
\end{proof}

\begin{lem}\label{Reduced_Split_Edges}
For every $q, q'$, and $i$, it holds that $\zug{\rzug{q,i},\rzug{q',i+i}} \in
E'$ iff there are nodes $v$ and $v'$ in $T^{rs}$ so that $\abs{v}=i$, $v'$ is
a child of $v$, $q \in T^{rs}(v)$, and $q' \in T^{rs}(v')$. 
\end{lem}
\begin{proof}
This follows from Lemma \ref{Reduced_Split_Nodes} and Lemma \ref{Gprime_Captures_Profiles}.
\end{proof}

Finally, the \emph{skeleton} $T^{sp}$ is obtained by removing from the reduced
split tree all nodes that are finite. As a corollary of Lemma
\ref{Reduced_Split_Edges}, the skeleton is a representation of $\Gdubprime$.
The slice automaton of \kahler and Wilke proceeds by tracking the levels of
$T^{rs}$ and guessing which nodes occur in $T^{sp}$.  Each level $i$ of
$T^{rs}$ is encoded as a \emph{slice}, a sequence $\zug{P_0,\ldots, P_m}$ of
pairwise disjoint subsets of $Q$. This slice is precisely the sequence of
equivalence classes in level $i$ of \Gprime, indexed by their relative
lexicographic ordering (see Figure~\ref{Fig:GDubPrime}).  The automaton of
\kahler and Wilke differs from Definition~\ref{Slice_Def} only in the details
of labeling states and the cut-point construction.

\end{document}